\documentclass[10pt,journal,compsoc]{IEEEtran}

\usepackage{cite}
\usepackage{amsbsy,amsmath,amssymb,epsfig,bbm,mathrsfs,multirow,amsthm,amsfonts}
\usepackage{algorithm,algorithmic}
\usepackage{graphicx}
\usepackage{textcomp}
\usepackage{multicol}
\usepackage{url}

\usepackage[table]{xcolor}

\graphicspath{{./figures/}}

\def\BibTeX{{\rm B\kern-.05em{\sc i\kern-.025em b}\kern-.08em
    T\kern-.1667em\lower.7ex\hbox{E}\kern-.125emX}}

\newtheorem{proposition}{Proposition}

\newtheorem{definition}{Definition}

\usepackage[T1]{fontenc}
\usepackage{makecell}
\usepackage{amstext} 
\usepackage{subcaption} 
\definecolor{mygray}{gray}{0.6}
\definecolor{myblue}{rgb}{0.8,0.85,1}

\usepackage{array, tabularx, boldline}
\newcolumntype{L}[1]{>{\raggedright\let\newline\\\arraybackslash\hspace{0pt}}m{#1}}
\newcolumntype{C}[1]{>{\centering\let\newline\\\arraybackslash\hspace{0pt}}m{#1}}
\newcolumntype{R}[1]{>{\raggedleft\let\newline\\\arraybackslash\hspace{0pt}}m{#1}}

\usepackage{cellspace}
\begin{document}

\title{A Game-theoretic Approach Towards Collaborative Coded Computation Offloading}

\author{
%Jer Shyuan Ng\thanks{JS.~Ng and WYB.~Lim are with Alibaba Group and Alibaba-NTU Joint Research Institute, Nanyang Technological University, Singapore. }, 
%Wei Yang Bryan Lim, 
%%Nguyen Cong Luong\thanks{N.~C.~Luong is with Faculty of Information Technology, PHENIKAA University, Hanoi 12116, Vietnam, and is with PHENIKAA Research and Technology Institute (PRATI), A\&A Green Phoenix Group JSC, No.167 Hoang Ngan, Trung Hoa, Cau Giay, Hanoi 11313, Vietnam.},
%Zehui Xiong\thanks{Z.~Xiong is with Alibaba-NTU Joint Research Institute, and also with School of Computer Science and Engineering, Nanyang Technological University, Singapore. }, 
%%Alia Asheralieva,\thanks{A.~Asheralieva is with Department of Computer Science and Engineering, Southern University of Science and Technology, Shenzhen, China.}
%Dusit~Niyato,\thanks{D.~Niyato is with School of Computer Science and Engineering, Nanyang Technological University, Singapore. }~\textit{IEEE~Fellow}, \\
%Cyril Leung,\thanks{C. Leung is with The University of British Columbia and Joint NTU-UBC Research Centre of Excellence in Active Living for the Elderly (LILY).}
%%Chunyan Miao \thanks{C.~Miao is with Joint NTU-UBC Research Centre of Excellence in Active Living for the Elderly (LILY) and School of Computer Science and Engineering, Nanyang Technological University, Singapore.}
%Junshan Zhang,\thanks{J.~Zhang is with the School of Electrical, Computer and Energy Engineering, Arizona State University, USA.} 
%Qiang Yang,\thanks{Q.~Yang is with Hong Kong University of Science and Technology, Hong Kong, China.} 
%}

Jer Shyuan Ng, Wei Yang Bryan Lim, Zehui Xiong, Dusit~Niyato,~\IEEEmembership{Fellow,~IEEE}, Cyril Leung, \\Dong In Kim,~\IEEEmembership{Fellow,~IEEE}, Junshan Zhang,~\IEEEmembership{Fellow,~IEEE}, Qiang Yang,~\IEEEmembership{Fellow,~IEEE}%
\IEEEcompsocitemizethanks{\IEEEcompsocthanksitem JS.~Ng and WYB.~Lim are with Alibaba Group and Alibaba-NTU Joint Research Institute, Nanyang Technological University, Singapore.
\IEEEcompsocthanksitem Z.~Xiong is with Pillar of Information Systems Technology and Design, Singapore University of Technology Design.
\IEEEcompsocthanksitem D.~Niyato is with School of Computer Science and Engineering, Nanyang Technological University, Singapore. 
\IEEEcompsocthanksitem C. Leung is with The University of British Columbia and Joint NTU-UBC Research Centre of Excellence in Active Living for the Elderly (LILY). 
\IEEEcompsocthanksitem  DI.~Kim is with Sungkyunkwan University, South Korea.
\IEEEcompsocthanksitem J.~Zhang is with School of Electrical, Computer and Energy Engineering, Arizona State University, USA. 
\IEEEcompsocthanksitem Q.~Yang is with Hong Kong University of Science and Technology, Hong Kong, China.
}
}
%\author{Jer Shyuan Ng$^{1,2}$, Wei Yang Bryan Lim$^{1,2}$, Sahil Garg$^{3}$, Zehui Xiong$^{2,4}$,\\ Dusit Niyato$^{4}$, Mohsen Guizani$^{5}$, and Cyril Leung$^{6,7}$\\
%$^1$Alibaba Group~$^2$Alibaba-NTU JRI~$^3$École de Technologie Supérieure~$^4$SCSE, NTU, Singapore~$^5$Qatar University~\\$^6$LILY Research Center, NTU, Singapore~$^7$ECE, UBC, Canada \vspace*{-5mm}}

\IEEEtitleabstractindextext{%
\begin{abstract} 

%As the amount of data collected for IoT applications increases rapidly due to improved sensing capabilities and the increasing number of Internet of Things (IoT) devices, a single cloud server is no longer able to handle the large-scale datasets individually. Given the improved computational capabilities of the edge servers such as the base stations and wireless access points, 
Coded distributed computing (CDC) has emerged as a promising approach because it enables computation tasks to be carried out in a distributed manner while mitigating straggler effects, which often account for the long overall completion times. Specifically, by using polynomial codes, computed results from only a subset of edge servers can be used to reconstruct the final result. However, incentive issues have not been studied systematically for the edge servers to complete the CDC tasks. In this paper, we propose a tractable two-level game-theoretic approach to incentivize the edge servers to complete the CDC tasks. Specifically, in the lower level, a hedonic coalition formation game is formulated where the edge servers share their resources within their coalitions. By forming coalitions, the edge servers have more Central Processing Unit (CPU) power to complete the computation tasks. In the upper level, given the CPU power of the coalitions of edge servers, an all-pay auction is designed to incentivize the edge servers to participate in the CDC tasks. In the all-pay auction, the bids of the edge servers are represented by the allocation of their CPU power to the CDC tasks. %All base stations submit their bids regardless of whether they win or lose in the auction. The edge servers, given their coalitions of workers formed, compete to allocate more CPU power for the CDC tasks in order to win the rewards offered by the cloud server. 
The all-pay auction is designed to maximize the utility of the cloud server by determining the allocation of rewards to the winners. Simulation results show that the edge servers are incentivized to allocate more CPU power when multiple rewards are offered, i.e., there are multiple winners, instead of rewarding only the edge server with the largest CPU power allocation. Besides, the utility of the cloud server is maximized when it offers multiple homogeneous rewards, instead of heterogeneous rewards.

\end{abstract}

}

\maketitle

\IEEEraisesectionheading{\section{Introduction}}

%The enhanced sensing, computational and communication capabilities of devices at the edge of the Internet have led to the ubiquitous adoption of the Internet of Things (IoT). In contrast to the traditional sensor devices that lack of computational and communication capabilities, the IoT devices are connected to each other over the Internet and mainly consumer-centric, i.e., the data generated is related to the users themselves or their surroundings, thereby empowering the development of many useful IoT applications. 

%Different from the traditional sensor devices that lack of computational and communication capabilities, the IoT devices are able to generate data about the environment and share the data by uploading it to the Internet.  

\IEEEPARstart{C}{oupled} with reliable wireless communication technologies, IoT devices can serve as important sources of sensor data for Artificial Intelligence (AI) technologies to be leveraged, towards the development of data-driven applications~\cite{liu2018survey}. In particular, many machine learning models are developed to monitor various large-scale physical phenomena for smart city applications, such as prediction of road conditions~\cite{kalim2016crater}, air quality monitoring~\cite{dutta2017towards} and tracking of medical conditions~\cite{jovanovic2019mobile}. Edge computing~\cite{shi2016promise} has emerged as a promising approach that extends cloud computing services to the edge of the networks. In particular, by leveraging on the computational capabilities, e.g., Central Processing Unit (CPU) power, of the edge servers, e.g., base stations and edge devices, e.g., laptops and tablets, the cloud server can offload its computation tasks to the edge servers and devices.%\footnote{The computation tasks can also be extended to edge devices such as vehicles and laptops.}. 

However, there are several challenges pertaining to the distributed edge computing network that need to be addressed for efficient and scalable implementation. Firstly, since several edge servers perform the distributed computation tasks collaboratively, the communication costs can be high due to the frequent exchange of intermediate results. Secondly, the response times vary across the edge servers due to several factors such as imbalanced work allocation, contention of shared resources and network congestion~\cite{dean2013tail,ananthanarayanan2010reining}. Thirdly, the confidentiality of the data may be compromised as eavesdroppers may monitor data transmission over wireless channels.

Coded distributed computing (CDC)~\cite{ng2020survey} has been proposed as an efficient method for distributed computation tasks at the edge of the network. In particular, coding techniques are used to design computation strategies that divide the entire dataset and allocate subsets of data to the edge servers for computations. In the distributed edge computing network, one of the main challenges is the straggler effects where the task completion time is determined by the slowest edge server as the cloud server needs to wait for all edge servers to return their results before it can reconstruct the final result. As a result, the latency of the distributed computation tasks can be high~\cite{aktas2018relaunch, aktas2017clones}. %Although CDC schemes are able to overcome the different challenges of distributed computing such as communication inefficiency and privacy loss, we focus on the use of CDC schemes to mitigate the straggler effects. 
By using CDC schemes\footnote{CDC schemes do not only mitigate straggler effects, but can also reduce communication costs and ensure security in the distributed edge computing network. This paper focuses on CDC schemes that aim to mitigate straggler effects.}, instead of having to wait for all edge servers to complete their computation tasks, the cloud server only needs to wait for a subset of edge servers to return their results. Hence, CDC schemes can reduce computation latency by obviating the need to wait for the slower edge servers.

%Therefore, we adopt the use of coding techniques to mitigate the straggler effects so that the cloud server only needs the computed results from a subset of IoT devices, instead of all devices, thereby leads to a shorter overall time needed to complete the distribute computation tasks. 

However, incentives are essential for the edge servers to participate in or to complete their allocated CDC subtasks. To design an appropriate incentive mechanism, it is important to consider the unique characteristics of the CDC framework. Specifically, even though the edge servers are each allocated a subset of the entire dataset for computations, some of the edge servers' computed results may not be used to reconstruct the final result, e.g., due to straggling. These edge servers in turn do not receive any compensation. As a result, this may discourage the participation of certain edge servers. To address this challenge, we propose an all-pay auction to model the competition between the different edge servers and at the same time, improve the participation of edge servers so as to elicit more CPU power for the CDC tasks. 

In distributed edge computing networks, the edge servers may work together with various edge devices, by forming coalitions in order to complete their computation tasks. To model the cooperation between the edge servers and devices, we propose a hedonic coalition formation game in which the edge devices decide which edge server to join based on their utility-maximizing objectives. In analogy to practical scenarios, the edge devices make decisions that maximize their utilities without taking into consideration the effect of their decisions on other edge servers or devices. 

The main aim of this work is to develop an incentive mechanism for enabling efficient completion of CDC tasks for IoT applications. %The design of theincentive mechanism for CDC tasks is important as it incentivizes edge servers and devices to allocate more CPU power for the computation tasks by offering rewards for their contributions, while straggler effects are mitigated through the use of efficient coding techniques. 
Our key contributions are summarized as follows:
\begin{enumerate}

%\item We introduce the concept of CDC to perform distributed computation tasks so as to reduce straggler effects.

\item We highlight the importance of incentives in CDC, which is an issue ignored, but crucial toward economically sustainable distributed systems, by existing works.

\item We propose a two-level game theoretic approach to incentivize the edge servers to contribute their CPU power for the CDC tasks. 

%\item We formulate the problem of coalition formation of edge servers and devices as a hedonic coalition formation game in which each edge device acts selfishly and aims to maximize its own utility.

\item We formally show that the edge servers may improve their utilities by forming coalitions. We, therefore, introduce a hedonic coalition formation game to achieve a stable coalitional structure.

\item We adopt an all-pay auction to model the competition between the different edge servers (with their coalitions of edge devices) which aim to win the rewards offered by the cloud server and analyze the different reward structures that affect the utility of the cloud server. 

\item We evaluate the performance of the proposed scheme. Simulation results show that the total amount of CPU power allocated for the CDC tasks is higher under the proposed scheme as compared to random CPU power allocation.

\end{enumerate}

%In the all-pay auction which rewards the cluster heads based on the CPU power contributions, the abilities of the cluster heads to win the rewards offered by the master are limited by their own CPU power. In other words, even if the cluster head contributes all of its CPU power for the allocated CDC subtask, it may not able to win the largest amount of reward if there are other cluster heads that can allocate larger amount of CPU power. As such, the cluster heads may cooperate with the wireless access points, i.e., workers, to increase the amount of available CPU power for the allocated CDC subtasks. 

The remainder of the paper is organized as follows. Section~\ref{sec:related} highlights the related works. Section~\ref{sec:system} presents the system model and problem formulation. Section~\ref{sec:lower} and Section~\ref{sec:upper} discuss the hedonic coalition formation game and the design of an all-pay auction respectively. Section~\ref{sec:simulate} reports the simulation results and analysis of the proposed two-level game-theoretic approach. Section~\ref{sec:conclude} concludes the paper.

\section{Related Work}
\label{sec:related}

We discuss the recent studies related to three different areas, i.e., (i) coded distributed computing, (ii) coalitional formation game, and (iii) auction design.

\subsection{Coded Distributed Computing (CDC)}

Given the emergence of big data which necessitates computation- and storage-intensive processing, large-scale distributed systems have received significant attention from both the research and industrial communities. A number of studies in the literature have focused on the minimization of communication load of the distributed computation tasks. Network coding in the context of distributed cache systems has been a promising approach to increase network throughput and improve performance by jointly optimizing data placement and delivery phases~\cite{maddah2014fundamental,karam2016hierarchical}. 

Recently, coding techniques have increasingly been used in distributed computing networks. One of the active research areas is the minimization of the communication load in the data shuffling phase through coded multicast transmission as this phase accounts for a large proportion of the overall execution time~\cite{zhang2013performance}. There is a tradeoff between computation load and communication load~\cite{li2018tradeoff}. In order to reduce the number of communication rounds, which is significant for distributed iterative algorithms, \cite{haddadpour2018cross} proposes a computing technique that jointly codes the computation at multiple iterations by leveraging on the storage and computation redundancy of the workers. The work in~\cite{wan2020topological} considers the network topology of the distributed systems in designing an efficient CDC scheme for practical implementation. It relaxes the assumption that the physically-separated servers are connected to a single error-free common communication bus.

Apart from the studies that focus on the minimization of communication load in the distributed computing networks, coding techniques are also used to alleviate the stragglers' delays that limit the performance as distributed computing systems are scaled up. This is achieved by reducing the recovery threshold. Various CDC schemes are proposed for different computation problems, e.g., matrix multiplication~\cite{yu2017polynomial,lee2017high}, gradient descent~\cite{raviv2019gradient}, convolution~\cite{dutta2017coded}, linear transform~\cite{wang2018fundamental} and Fourier transform~\cite{yu2017fourier}. Instead of ignoring the partial computations that are completed by the stragglers, several studies such as~\cite{kiani2018exploitation} and~\cite{ozfatura2019speeding} exploit the work completed by the stragglers through sequential processing and multi-message communication. In~\cite{dai2020sazd}, the computation load is reduced by removing complex multiplication and division operations in the encoding and decoding phases. 

However, to the best of our knowledge, there are few studies that focus on the design of incentive mechanisms for CDC tasks. Given the interactions of autonomous and non-cooperative agents in the networks, an effective incentive mechanism design is an important step towards realizing the scalable and efficient implementation of CDC schemes in distributed edge computing networks.

%An effective incentive mechanism is useful for the practical implementation of the CDC schemes by modeling the interactions between different agents in the network and thus, incentivizing them to facilitate the coded distributed computation tasks collaboratively.

\subsection{Coalitional Formation Game}

Due to the limited resources of a single device in completing the allocated task individually, coalition formation games in computation offloading have been investigated. In~\cite{ennya2018computing}, the fogs can cooperate with each other by sharing their resources in order to offer better quality of service and experience for the users. A joint coalition-and-pricing based data offloading framework is proposed in~\cite{zhang2018data} to maximize the data throughput and determine the equilibrium prices that promote cooperation between devices and the edge servers. Different from the generic coalitional formation games where the utilities of the players depend on the coalitional structures, the hedonic coalitions are formed based on the individual preference of the players. As such, the utilities of the players depend solely on the members of the coalitions to which the worker belongs. In~\cite{wahab2018trust}, a trust-based hedonic coalitional game is formulated to model the formation of trustworthy multi-cloud communities that are resilient to collusion attacks by malicious devices. 

It is often assumed that the resources of a coalition are dedicated for a particular computation task. In practical scenarios, each coalition may be required to complete multiple computation tasks. Hence, in order to incentivize the edge servers to allocate more resources to complete the CDC tasks, we adopt an auction scheme.

%However, these coalition formation approaches do not consider the rewards offered for facilitating the coded distributed computation tasks, which may affect the value of the coalitions formed. Besides, the competition between different coalitions for the rewards offered by the cloud server is not considered. To that end, we adopt an auction scheme with different reward structures that maximize the utility of the cloud server. 

\subsection{Auction Design}

The task of auction design for optimal allocation of resources and tasks is well-explored in the literature. In particular, in crowdsensing applications where the usefulness of the applications depends on the quantity and quality of data, auction theory is one of the important tools to achieve mutual agreement between the crowd-sourcer and the users. Specifically, an all-pay auction is used to encourage the contributions of users, e.g., data, that are used to solve a crowd-sourcing problem. In an all-pay auction, not all users that contribute to the task defined by the crowd-sourcer are rewarded. This is similar to the blockchain model illustrated in~\cite{xu2020dynamic} where only the miner which successfully generates a new block is rewarded.

In the literature, the design of the all-pay auctions considers different objectives and approaches. For example, some studies focus on the maximization of the quality of the contributions~\cite{archak2009optimal} while others focus on the maximization of the sum of the contributions~\cite{yoon2012optimal, ng2020collaborative}. The work in~\cite{xiao2016asymmetric} studies the total expected performance of asymmetric players in competing for heterogeneous prizes under a complete-information setting. In contrast, the study of~\cite{tie2014optimal} considers an incomplete information setting where users do not know how other users value the reward offered by the crowd-sourcer. Besides, the crowd-sourcer maximizes its profit by rewarding the winner based on its contribution. Several studies such as~\cite{cohen2008allocation} and~\cite{wen2016optimal} have analyzed the optimal prize structures for crowdsensing platforms. 

However, the formation of coalitions in auctions has seldom been considered. Here, we adopt a game-theoretic approach to incentivize the edge servers to contribute their CPU power for the CDC tasks.

\section{System Model and Problem Formulation}
\label{sec:system}

\begin{figure}[!t]
\includegraphics[width=\linewidth]{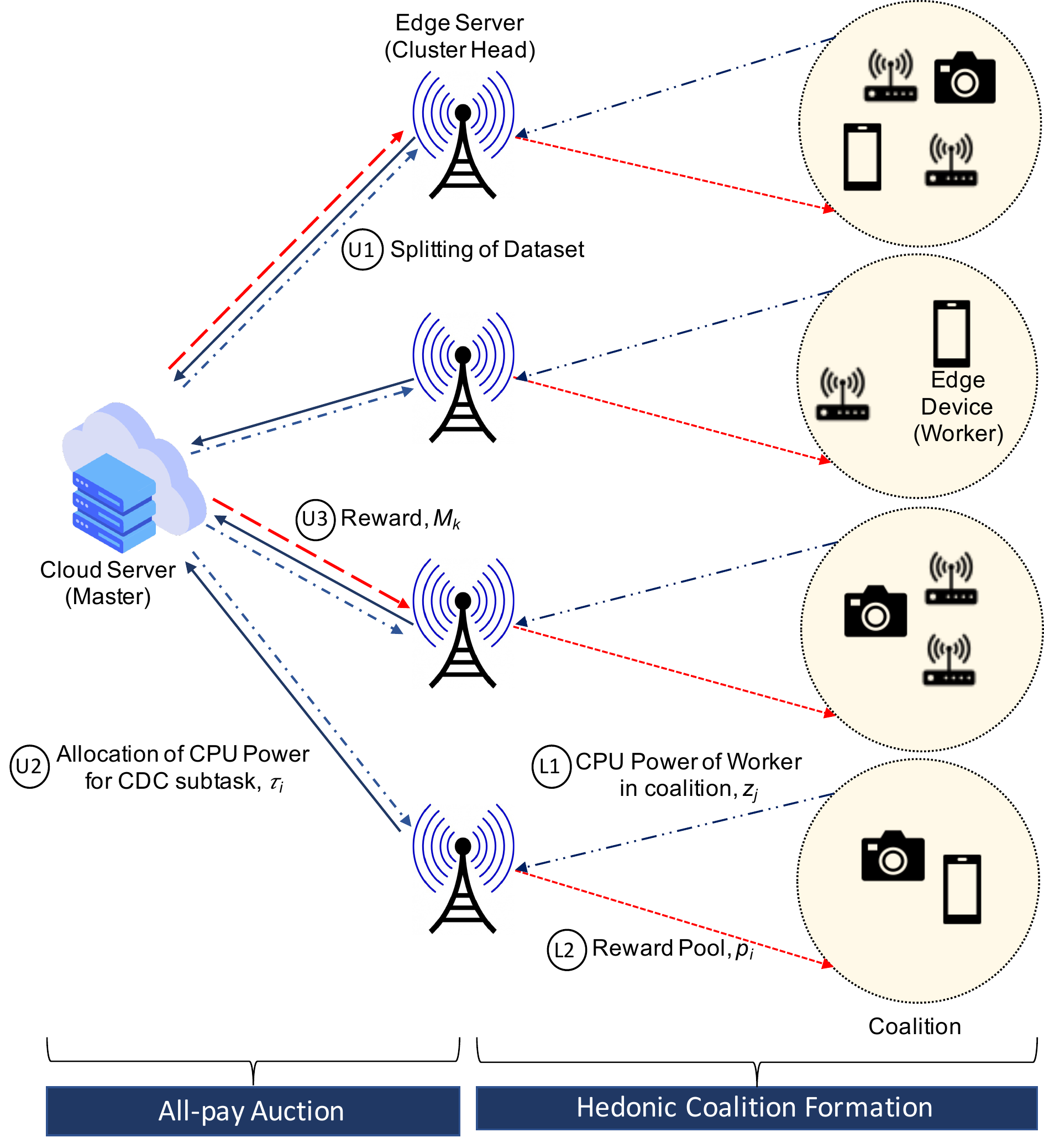}
\caption{\footnotesize{System model consists of the cloud server (master), edge servers (cluster heads) and edge devices (workers). In the lower level, there are two steps: (L1) the workers in the coalitions allocate their CPU power, and (L2) the cluster heads offer a reward pool to the workers in the coalitions. In the upper level, there are three steps: (U1) the master splits the dataset using polynomial codes, (U2) the cluster heads allocate CPU power for the CDC subtasks, and (U3) the cluster heads are rewarded for completing the allocated CDC subtasks. }}
\label{system}
\end{figure}

\begin{table}[t]
\caption{System Model Parameters.} 
\label{tab:notations}
\centering
%\scriptsize
\begin{tabularx}{8.7cm}{|Sl|X|}
%\begin{tabular}{|p{1.5cm} |p{6.4cm}| }

\hline
\rowcolor{lightgray}
\textbf{Parameter}& \textbf{Description}\\  \hline

$I$ & Number of cluster heads\\ \hline
$J$ & Number of workers\\ \hline
$K$ & Recovery threshold\\ \hline
$\rho_i$ & Reward pool by cluster head\\ \hline
$v(S_i)$ & Value of coalition\\ \hline
$x_j^{S_i}$ & Utility of worker $j$ in coalition $S_i$\\ \hline
$\sigma$ & Total amount of reward\\ \hline
$M_k$ & Size of reward\\ \hline
$\tilde{\mathbf{A}}_i$, $\tilde{\mathbf{B}}_i$ & Allocated matrices to cluster heads\\ \hline
$\tilde{\mathbf{C}}_i$ & Computed results by cluster heads\\ \hline
$\pi$ & Expected utility of master\\ \hline
$z_i$ & CPU power of cluster head\\ \hline
$z_j$ & CPU power of worker\\ \hline
$\mu_{ij}$ & Communication cost between worker $j$ and cluster head $i$\\ \hline
$\kappa$ & Effective switch coefficient\\ \hline
$a_i$ & Total number of CPU cycles\\ \hline
$\theta^{p}$ & Unit cost of computational energy\\ \hline
$\theta^{c}$ & Unit cost of communication energy\\ \hline
$c_i$ & Communication energy of cluster head $i$\\ \hline
$\tau_i$ & Allocated CPU power for the CDC subtasks\\ \hline
$\alpha_i$ & Utility of cluster head\\ \hline
$u_i$ & Expected utility of cluster head\\ \hline
$u_j(S_i)$ & Preference function of worker $j$ in coalition $S_i$\\ \hline
$v_i$ & Valuation of cluster head for total reward\\ \hline
$p_i^k$ & Probability of winning the reward\\ \hline

\end{tabularx}
\end{table}

\subsection{System Setting}

We consider a heterogeneous distributed edge computing network as illustrated in Fig.~\ref{system}. The system model consists of a master, i.e., cloud server, and a set $\mathcal{I}=\{1,\ldots,i,\ldots,I\}$ of $I$ cluster heads, e.g., edge servers, that have different computational capabilities and belong to different service providers. Moreover, there are $J$ workers, e.g., edge devices, represented by the set $\mathcal{J}=\{1,\ldots, j,\ldots,J\}$, that also have different computational capabilities to facilitate in the computation tasks. In IoT networks, for example, the ubiquity of the IoT devices as well as their on-board sensing and processing capabilities are leveraged to collect data for many innovative IoT applications. Given the large amounts of sensor data collected from different IoT devices, the master aims to perform the training of an AI model to complete a user-defined data processing task. As the number of IoT devices increases, so does the size of the dataset that the master needs to handle. However, the master may not have sufficient resources, i.e., computation power, to handle the growing dataset. Instead, it may utilize the resources of the cluster heads to complete the computation tasks in a distributed manner. The cluster heads may cooperate with the workers to increase their capabilities in completing the computation tasks. In particular, more CPU power can be allocated for the computation tasks.

\subsection{Coded Distributed Computing (CDC)}
%We consider that a resource-constrained cloud server, i.e., master, aims to build a prediction model by implementing the CDC schemes over the distributed edge computing network. 

One of the main challenges in performing distributed computation tasks is the straggler effect. In order to reduce the computation latency of the distributed computation tasks, the master applies CDC schemes over the distributed edge computing network. Coding techniques such as polynomial codes~\cite{yu2017polynomial} can be used to mitigate straggler effects by reducing the recovery threshold, i.e., the number of cluster heads that need to submit their results for the master to reconstruct the final result. In order to perform coded distributed matrix multiplication computations, i.e., $\mathbf{C}=\mathbf{A}^\top \mathbf{B}$ where $\mathbf{A}$ and $\mathbf{B}$ are input matrices\footnote{The matrix multiplication may also involve more than two matrices. Our system model can be easily extended to solve the matrix multiplication of more than two matrices.}, $\mathbf{A} \in \mathbb{F}_q^{s\times r}$ and $\mathbf{B} \in \mathbb{F}_q^{s\times t}$ for integers $s$, $r$, and $t$ and a sufficiently large finite field $\mathbb{F}_q$, there are four important steps:

\begin{enumerate}
\item \emph{Task Allocation: }Given that all cluster heads are able to store up to $\frac{1}{m}$ fraction of matrix $\mathbf{A}$ and $\frac{1}{n}$ fraction of matrix $\mathbf{B}$, the master divides the input matrices into submatrices $\tilde{\mathbf{A}}_i=f_i(\mathbf{A})$ and $\tilde{\mathbf{B}}_i=g_i(\mathbf{B})$, where $\tilde{\mathbf{A}}_i \in \mathbb{F}_q^{s\times \frac{r}{m}}$ and $\tilde{\mathbf{B}}_i \in \mathbb{F}_q^{t\times \frac{r}{n}}$ respectively. Specifically, $\mathbf{f}$ and $\mathbf{g}$ represent the vectors of functions such that $\mathbf{f}=(f_1,\ldots,f_i,\ldots,f_I)$ and $\mathbf{g}=(g_1,\ldots,g_i,\ldots,g_I)$, respectively. Then, the master distributes the submatrices to the cluster heads over the wireless channels for computations.

\item \emph{Local Computation: }Each cluster head $i$ is allocated submatrices $\tilde{\mathbf{A}}_i$ and $\tilde{\mathbf{B}}_i$ by the master. Based on the allocated submatrices, the cluster heads perform matrix multiplication, i.e., $\tilde{\mathbf{A}}_i^\top \tilde{\mathbf{B}}_i$, $\forall i \in \mathcal{I}$.

\item \emph{Wireless Transmission: }Upon completion of the local computations, each cluster head transmits its computed results, i.e., $\tilde{\mathbf{C}}_i=\tilde{\mathbf{A}}_i^\top \tilde{\mathbf{B}}_i$ to the master over the wireless communication channels.

\item \emph{Reconstruction of Final Result: }By using coding techniques, the master is able to reconstruct the final result upon receiving $K$ out of $I$ computed results by using decoding functions. In other words, the master does not need to wait for all $I$ cluster heads to complete their allocated CDC subtasks. Note that although there is no constraint on the decoding functions to be used, a low-complexity decoding function such as the Reed-Solomon decoding algorithm~\cite{didier2009efficient} ensures the efficiency of the overall matrix multiplication computations.
\end{enumerate}

By using the polynomial codes~\cite{yu2017polynomial}, the optimum recovery threshold that can be achieved where each cluster head is able to store up to $\frac{1}{m}$ of matrix $\mathbf{A}$ and $\frac{1}{n}$ of matrix $\mathbf{B}$ is defined as:
\begin{equation}
K=mn.
\end{equation}

The training of an AI model may involve various types of distributed computation problems, e.g., matrix multiplication, stochastic gradient descent, convolution and Fourier transform. Without loss of generality, we consider the distributed matrix multiplication computations. Matrix multiplication is an important operation underlying many data analytics applications, e.g., machine learning, scientific computing and graph processing~\cite{yu2017polynomial}. %Due to the lack of computational capabilities and storage resources, the resource-constrained cloud server employs multiple base stations, i.e., cluster heads, to facilitate the distributed computation task. 
%First, the master uses polynomial codes~\cite{yu2017polynomial} to determine the algebraic structure of the encoded submatrices and distributes the submatrices to the base stations, i.e., cluster heads. Then, the cluster heads perform local computations based on their allocated datasets, and return the computed results to the master. The master reconstructs the final result by using fast polynomial interpolation algorithms such as the Reed-Solomon decoding algorithm~\cite{didier2009efficient}. 

However, there needs to be an incentive for a cluster head to be one of the $K$ cluster heads to complete their local computations of CDC subtasks and return their computed results to the master. 

\subsection{Two-level Game-theoretic Approach}
In this paper, we focus our study on a two-level game-theoretic approach as follows: (i) in the lower level, we adopt a hedonic coalition formation game to investigate the coalition formation of workers to facilitate the computation tasks of the cluster heads, and (ii) in the upper level, we study the all-pay auction to encourage the cluster heads, given the coalitions of workers formed, to allocate more CPU power for the CDC subtasks while maximizing the utility of the master. It is assumed that the cluster heads do not consider forming coalitions among themselves since they are independent and competing service providers.

%The cluster heads may not want to participate in the distributed computation tasks or to perform to the best of their abilities. For example, the cluster heads may not allocate large amount of CPU power for the coded distributed tasks, resulting in longer time needed to complete the computation tasks. In order to incentivize the cluster heads to allocate more CPU power to complete the allocated computation tasks, we adopt an all-pay auction approach to determine the rewards to the cluster heads while maximizing the utility of the master. Specifically, the master offers $K$ rewards for the first $K$ workers that contribute the largest amount of CPU power.

%By performing computations on the allocated submatrices individually, the CPU power allocated by the cluster heads are limited by their own computational capabilities, e.g., the number of cores the processors have. As a result, the cluster heads may not be able to win the largest amount of reward, i.e., top reward, even if they contribute the entire amount of their CPU power to complete the allocated subtasks. In order to improve the chances of the cluster heads in winning the top reward offered by the master, the cluster heads can consider to form coalitions with the workers. Each worker $j$ can choose to join any cluster head $i \in \mathcal{I}$. Note that each worker is only allowed to choose to facilitate the computation task of one of the cluster heads.

\subsubsection{Lower-level Hedonic Coalition Formation}

Given $I$ cluster heads and $J$ workers in the network, the formation of the coalitions is derived in the lower level. In order to encourage more workers to facilitate its computation tasks, each cluster head offers a reward pool to the coalition of workers. The reward pools for the cluster heads maybe different depending on their available budgets. The reward that each worker receives is a function of its proportion of CPU power contributed in the coalition, which, for example, can be measured from the computation latency of that worker. On the one hand, workers are incentivized to join a cluster head that has a greater reward pool in the hope of receiving a higher reward. On the other hand, as more workers join a cluster head, each worker will receive a smaller proportion of the reward pool as the pool needs to be shared among more workers. In addition to the amount of rewards, the workers' utilities are also affected by its computation and communication costs. Hence, the workers make their decisions based on their utilities. Each worker $j$ can choose to join any cluster head $i \in \mathcal{I}$. Note that each worker is only allowed to choose to facilitate the computation tasks of one of the cluster heads. In practice, worker $j$ may be limited in the choice of cluster heads it can join, e.g., due to geographical location.

%The CPU power allocated by the cluster heads are limited by their own computational capabilities, e.g., the number of cores the processors have. As a result, the cluster heads may not be able to win the largest amount of reward, i.e., top reward, even if they contribute the entire amount of their CPU power to complete the allocated CDC subtasks. %The objective of each cluster head is to win the top reward from the master. The cluster heads can increase their chances of winning the top reward by attracting more workers to join their coalitions. With more workers joining the coalition and contributing their CPU power, the cluster heads can allocate larger amount of CPU power to complete the allocated CDC subtasks, thereby increasing their chances of winning the top reward. 

%In order to encourage more workers to facilitate its CDC subtask, each cluster head offers a reward pool to the coalition of workers. Each cluster head offers different amount of reward pool as they have different willingness to share their rewards. In other words, the cluster heads retain a proportion of the reward offered by the master and share the rest by rewarding the workers in the coalition. 

\subsubsection{Upper-level All-pay Auction}

The lower-level coalition formation game determines the amount of CPU power that each coalition has. The coalitions with greater CPU power are more valuable to the master as they are able to complete the CDC subtasks within a shorter period of time. %However, there is no incentive for the cluster heads to allocate their entire CPU power for the allocated CDC subtasks as the cluster heads may need to complete other computation tasks. In order to attract the participation of cluster heads with high CPU power, the master offers rewards to the cluster heads. 
Since a cluster head with its coalition of workers may need to work on several computation tasks simultaneously, they may not allocate all their CPU power for CDC subtasks. In order to incentivize the cluster heads to allocate more CPU power to complete the CDC subtasks, the master offers rewards to the cluster heads. Since computed results are required from only a subset of cluster heads, the cluster heads need to compete for the rewards. In particular, we explore an all-pay auction mechanism whereby the cluster heads bid for the rewards. In this all-pay auction, although all cluster heads allocate CPU power to perform the computations on the allocated dataset, only $K$ cluster heads are rewarded. As such, the all-pay auction is designed such that the utility of the master is maximized by incentivizing the cluster heads to allocate more CPU power for the CDC subtasks.

In traditional auctions such as first-price and second-price auctions, only the winners of the auctions pay. In contrast, in all-pay auctions, regardless of whether the bidders win or lose, they are required to pay to participate in the auction. In this all-pay auction, the bids of the edge servers are represented by their CPU power, i.e., the number of CPU cycles, allocated by the edge servers to complete the CDC subtasks. In other words, the larger the CPU power allocated, the higher the bid of the edge server. 

There are two advantages of an all-pay auction~\cite{luo2016incentive}. Firstly, it reduces the probability of non-completion of allocated subtasks, thus allowing the cloud server to reconstruct the final result. This differs from traditional auctions in which the winners of the auctions can still choose not to complete their tasks and give up the reward that is promised by the auctioneer (cloud server). As a result, the auctioneer needs to conduct another round of auction. Secondly, it reduces the coordination cost between the auctioneer and the bidders (edge servers). Specifically, in traditional auctions, the participants need to bid then contribute whereas in all-pay auctions, the bids of the participants are directly determined by their contributions. In other words, the participants do not need to bid explicitly in all-pay auctions. This is particularly useful for the development of a scalable network since the communication overheads are reduced.

\subsubsection{Interaction between Lower and Upper Levels}

In the lower level, the workers form coalitions to support the computation tasks of the cluster heads, increasing the capabilities of the cluster heads to complete their computation tasks by, for example, reducing computation latency or increasing computation accuracy. Given the coalitions of workers formed, the cluster heads need to allocate CPU power for their computation tasks. Without proper incentive mechanisms, the cluster heads may randomly allocate CPU power for their computation tasks, which is not optimal as it does not maximize the utilities of the cluster heads. In order to incentivize the cluster heads to allocate CPU power for the CDC subtasks, an all-pay auction is proposed in the upper level. The lower-level hedonic coalition formation game helps to improve the utilities of the cluster heads by allowing them to allocate their equilibrium CPU power in the upper-level all-pay auction. Specifically, without forming coalitions, the capabilities of the cluster heads are limited by their own CPU power, thus not allowing them to allocate their equilibrium CPU power for the CDC subtasks, which may be greater than their own CPU power. As such, the cluster heads may not win the reward offered by the master. Therefore, the two-stage game theoretic approach ensures that the utilities of the cluster heads are maximized by allocating CPU power for the CDC subtasks.

\section{Lower-level Hedonic Coalition Formation}
\label{sec:lower}

In this section, we formulate the problem of collaborative execution of computation tasks as a hedonic coalition formation game. 
To form a coalition, each cluster head broadcasts its intention to form a coalition to all workers in the network. Each cluster head $i$ offers a reward pool $\rho_i$ to the coalition of workers. To decide whether to join or leave a coalition, each worker also compares its utility in the current coalition and the utility of joining another coalition. If the utility of joining another coalition is higher, the worker leaves the current coalition and joins another coalition, hence forming a new coalitional structure. The coalitional structure is stable when no worker has incentive to change its current coalition.

%Each cluster head aims to win the top reward from the master. The allocation of rewards of the master is based on the CPU power ranks of the cluster heads. The larger the amount of CPU power allocated, the higher the rank of the cluster head, the larger the amount of reward is offered to the cluster head. The allocation of rewards to the cluster heads is explained further in detail in Section~\ref{sec:upper}. In order to increase their chances of winning the top reward, the cluster heads may form coalitions with the workers to have more CPU power to facilitate the allocated computation subtasks.

%If a cluster head wins a reward from the master, it is willing to share a proportion of it in rewarding the workers in the coalition. However, the cluster heads have different willingness to share the reward, $\rho_i$. In particular, greater $\rho_i$ means cluster head $i$ has greater willingness to share the reward. 

%Then, each worker first randomly chooses a cluster head to join, in which it evaluates its own utility by joining the coalition. Since the utilities of the workers depend on the share of CPU power contributions in the coalition, the workers may receive only a small proportion of the reward if it joins a coalition of bigger size, i.e., with more coalition members. As such, the workers may be better off by joining a coalition of smaller size. 

\subsection{Hedonic Coalition Formation Formulation}

We present the definitions for hedonic coalition formation formulation.

\begin{definition}
A coalition of workers is denoted by $S_i \subseteq \mathcal{J}$ where $i$ is the index of the cluster head.
\end{definition}

In particular, workers in coalition $S_i$ facilitate the computation tasks of cluster head $i$, $\forall i \in \mathcal{I}$.

\begin{definition}
A partition or coalitional structure is a set of coalitions that spans all workers in $\mathcal{J}$. The coalitional structure is represented by $\Pi=\{S_1,\ldots,S_i,\ldots,S_I\}$, where $S_i \cap S_{i'}=\emptyset$ for $i \neq i'$, $\bigcup^I_{i=1}S_i=\mathcal{J}$ and $I$ is the total number of coalitions in coalitional structure $\Pi$~\cite{saad2009coalitional}.
\end{definition}
$\mathcal{J}$ denotes the coalition of all workers, which is also known as the grand coalition. The formation of a grand coalition means that all workers facilitate the computation tasks of a single cluster head. A singleton coalition is a coalition that only contains a single worker where only a worker facilitates the computation tasks of a cluster head. Note that the total number of coalitions equals the number of cluster heads in the network. If there is no worker that is willing to facilitate the computation tasks of cluster head $i$, $\forall i \in \mathcal{I}$, the coalition $S_i$ associated with cluster head $i$ is represented by an empty set, $\emptyset$.

\begin{definition}
A coalitional structure $\Pi^*=\{S_1^*,\ldots, S_i^*,\ldots, S_I^*\}$ is a stable coalitional structure if no coalition $S_i^* \in \Pi$ has an incentive to change the current coalitonal structure $\Pi$ by merging with another coalition $S_{i'}^*$, $S_i^* \cap S_{i'}^*=\emptyset$ for $i \neq i'$, or splitting into smaller disjoint coalitions.
\end{definition}

The value of any coalition $S_i \in \Pi$ is the total amount of CPU power of both the cluster head and workers in the coalition. %The larger the proportion of total CPU power of the workers in the coalition, the higher the probability of winning the top reward, the greater the value of the coalition. 
The value of coalition $S_i$, which is denoted as $v(S_i)$, is expressed as follows:
\begin{equation}
v(S_i)=\sum_{j\in S_i}z_j+z_i,
\end{equation}
where $z_j$ and $z_i$ are the amount of available CPU power of worker $j$ and cluster head $i$, respectively.

%Each cluster head $i$, $\forall i \in \mathcal{I}$, has different willingness to share the reward, $\rho_i$. This means that the proportion of reward that is offered to coalition $S_i$ for facilitating the distributed computation task of cluster head $i$ is represented by $(1-\rho_i)$. 

Each cluster head $i$, $\forall i \in \mathcal{I}$, offers different amount of reward pool, $\rho_i$. Each worker $j$ in coalition $S_i$ receives a proportion of the reward pool offered by the cluster head, for which the coalition $S_i$ provides support. The amount of reward that each worker receives depends on its proportion of the CPU power in the coalition, which, for example, can be measured from the worker's computation latency. The greater the proportion of CPU power, the larger the amount of reward the worker receives. Specifically, the utility of worker $j$ in coalition $S_i$, %represented by $x_j^{S_i}$, of winning the top reward, that is represented by $M_1$, 
is denoted as follows:
\begin{equation}
\label{eqn:utilwork}
x_j^{S_i}=\frac{z_j}{\sum_{j \in S_i}z_j}\rho_i-\delta_{j}z_j-\mu_{ij},
\end{equation}
where $\delta_{j}$ is the unit cost of CPU power of worker $j$ and $\mu_{ij}$ is the communication cost for worker $j$ to reach cluster head~$i$. In particular, the utility of worker $j \in \mathcal{J}$ depends only on the members of the coalition that it belongs to. 

Based on its own utility, each worker $j \in \mathcal{J}$ needs to build its own preference over all possible coalitions that it can join, where each worker $j$ compares the utilities of joining different coalitions. As such, the concept of preference relation is introduced to illustrate the preference of each worker over all possible coalitions.

\begin{definition}
For any worker $j \in \mathcal{J}$, a preference relation $\succ_j$ is defined as a complete, reflexive and transitive binary relation over the set of all coalitions that worker $j$ can possibly join~\cite{anna2002stability}.
\end{definition}

The preference relation of worker $j \in \mathcal{J}$ can be expressed as follows:
\begin{equation}
S_1 \succeq_j S_2 \iff u_j(S_1) \geq u_j(S_2),
\end{equation}
where $S_1 \subseteq \mathcal{J}$ and $S_2 \subseteq \mathcal{J}$ are two possible coalitions that worker $j$ may join, $u_j(S_i)$ is the preference function for any worker $j \in \mathcal{J}$ and for any coalition $S_i$, $\forall i \in \mathcal{I}$.  In particular, for any worker $j \in \mathcal{J}$, given two coalitions $S_1 \subseteq \mathcal{J}$ and $S_2\subseteq \mathcal{J}$ where $j \in S_1$ and $j \in S_2$, $S_1 \succeq_j S_2$ means that worker $j$ prefers coalition $S_1$ over coalition $S_2$, or at least worker $j$ values both coalitions equally. In addition, its asymmetric counterpart, which is denoted as $\succ_j$, when used in $S_1 \succ_j S_2$ indicates that worker $j$ strictly prefers coalition $S_1$ over coalition $S_2$. It is worth noting that the preference relation, $\succeq_j$ is defined to allow the workers to quantify their preferences, which can be application-specific. The preference relation can be expressed as a function of several parameters such as the payoffs of workers in joining different coalitions and the proportion of the contribution of each worker in the same coalition.

The preference function of worker $j$ in coalition $S_i$ which is represented by $u_j(S_i)$ is defined as follows:
\begin{equation}
\label{eqn:function}
u_j(S_i)=
	\begin{cases}
   	x_j^{S_i}, & \text{if $S_i \notin h(j)$}, \\
   	-\infty, & \text{otherwise},
 	\end{cases}
\end{equation}
where $x_j^{S_i}$ is the utility of worker $j$ in coalition $S_i$ defined in Equation~(\ref{eqn:utilwork}) and $h(j)$ is the history set of worker $j$ that contains the list of coalitions that the worker $j$ has previously joined before the formation of the current coalitional structure $\Pi$. More specifically, the history set of worker $j \in \mathcal{J}$, $h(j)=\{S_{i_0}^{0},\ldots,S_{i_\lambda}^{\lambda},\ldots,S_{i_\Lambda}^{\Lambda}\}$, where $i_\lambda \in \mathcal{I}$, $j \in S_{i_\lambda}$ and $\Lambda$ represents the total number of changes in coalitions formed by worker $j$. Each time when a new coalition is formed, each worker $j \in \mathcal{J}$ updates its history set $h(j)$ by adding a new coalition $S_{i_\lambda}^{\lambda}$, where $i_\lambda \in \mathcal{I} $ and $j \in S_{i_\lambda}$.

As such, based on Equation~(\ref{eqn:function}), the preference of worker $j \in \mathcal{J}$ over the different coalitions is related to its utility defined in Equation~(\ref{eqn:utilwork}).

Given a set of workers $\mathcal{J}$ and a preference relation $\succeq_j$ for every worker $j \in \mathcal{J}$, a hedonic coalition formation game is formally defined as follows:
\begin{definition}
A hedonic coalition formation game is a coalitional game that is defined by $(\mathcal{J}, \succ)$ where $\mathcal{J}$ and $\succ=\{\succ_1,\cdots,\succ_j,\cdots,\succ_J\}$ represent the set of workers and the preference relation of each worker in $\mathcal{J}$ respectively. In addition, a hedonic coalition formation game fulfils the two important requirements as follows:

\begin{enumerate}
\item The payoff of any worker depends solely on the members of the coalitions to which the worker belongs, and
\item The coalition formed is a result of the preferences of the workers over the set of possible coalitions.
\end{enumerate}
\end{definition}

In the hedonic coalition formation game, based on the preference relation in Equation~(\ref{eqn:utilwork}) and preference function in Equation~(\ref{eqn:function}), the worker $j \in \mathcal{J}$ joins a new coalition that it has not visited before, and if and only if worker~$j$ achieves higher utility in the new coalition. Specifically, the formation of hedonic coalitions is based on switch rule, which determines whether the worker $j$ ($\forall j \in \mathcal{J}$) decides to leave or join a coalition.

\begin{definition}
\label{def:switch}
(Switch Rule) Given a coalitional structure $\Pi=\{S_1,\ldots,S_i,\ldots,S_I\}$, a worker $j$ decides to leave its current coalition $S_i$ and join another coalition $S_{i'} \in \Pi$, where $i \neq i'$, if and only if $S_{i'} \bigcup \{j\}\succ_{j} S_i$. As a result, $\{S_i, S_{i'}\} \rightarrow \{S_i \backslash \{j\}, S_{i'} \bigcup \{j\}\}$.
\end{definition}

The switch rule in hedonic coalition formation games allows any worker $j \in \mathcal{J}$ to leave its current coalition $S_i$ and join another coalition $S_{i'} \in \Pi$, where $i \neq i'$, given that $S_{i'} \bigcup \{j\}$ is strictly preferred over $S_{i}$ based on the defined preference relation. This transforms the current coalitional structure $\Pi$ into a new coalitional structure $\Pi'=\Pi \backslash \{S_i,S_{i'}\} \bigcup \{S_{i}\backslash \{j\}, S_{i'}\bigcup\{j\}\}$. The adoption of switch rule in hedonic coalition formation games reflects the selfish behaviour of the workers since the workers decide to leave and join any coalition based on their preference relations, without taking into account the effect of their actions on other workers.

\begin{proposition}
If worker $j$ performs the switch rule for the $\lambda$-th time, where it leaves its current coalition $S_{i_{\lambda-1}}^{\lambda-1}$ and forms a new coalition $S_{i_{\lambda}}^{\lambda}$, where $i_{\lambda-1} \neq i_{\lambda}$, the new coalition $S_{i_{\lambda}}^{\lambda}$ cannot be the same as any coalition formed in the history set, $h(j)$. In other words, before the update of the history set for the $\lambda$-th time, the new coalition is not in the history set, i.e., $S_{i_{\lambda}}^{\lambda} \notin h(j)$.
\end{proposition}

\begin{proof}

Suppose that there are two coalitions $S_{i_{\lambda-1}}^{\lambda-1}$ and $S_{i_{\lambda}}^{\lambda}$ such that $S_{i_{\lambda-1}}^{\lambda-1}=S_{i_\lambda}^{\lambda}$, where coalition $S_{i_{\lambda-1}}^{\lambda-1}$ is found in the history set, $h(j)$. Based on Equation~(\ref{eqn:utilwork}), the utility of worker~$j$ in either of the coalition is the same, i.e., $x_j^{S_{i_{\lambda-1}}^{\lambda-1}}=x_j^{S_{i_{\lambda}}^{\lambda}}$. According to the definition of switch rule in Definition~(\ref{def:switch}), worker $j$ only performs the switch operation if and only if the new coalition is strictly preferred over any of the previous coalitions. In other words, switch operation is only performed if and only if $x_j^{S_{i_\lambda}^{\lambda}}> x_j^{S_{i_{\lambda-1}}^{\lambda-1}}$. Since the newly formed coalition $S_{i_\lambda}^{\lambda}$ does not fulfil the condition of switch rule, the switch operation is not performed. More specifically, the newly formed coalition $S_{i_{\lambda}}^{\lambda}$ cannot be the same as any coalition $S_{i_\lambda}^{\lambda}$ in the history set, $h(j)$.

\end{proof}

In the formation of hedonic coalitions, there exists a stable coalitional structure. There are two types of stabilities of coalitional structure, i.e., Nash-stability and individual-stability~\cite{anna2002stability}. 

\begin{itemize}
\item \emph{Nash-stability:} A coalitional structure $\Pi=\{S_1,\ldots, S_i,\ldots, S_I\}$ is Nash-stable if no worker $j \in \mathcal{J}$ has incentive to leave its current coalition $S_i$ and join another coalition $S_{i'}$ where $i \neq i'$, i.e., $S_i \succ S_{i'} \bigcup \{j\}$, $\forall i' \in \mathcal{I}$. In other words, no worker is able to increase its utility by performing switch rule to change its current coalition.

\item \emph{Individual-stability:} A coalitional structure $\Pi=\{S_1,\ldots, S_i,\ldots, S_I\}$ is individually-stable if there does not exist such that (i) a worker $j$, $\forall j \in \mathcal{J}$ in its current coalition strictly prefers any other coalition, i.e., $S_{i'} \bigcup \{j\}\succ_j S_i$, $\forall i\ \in \mathcal{I}$, and (ii) the formation of a new coalition does not reduce the utilities of the members of the new coalition, i.e.,  $S_{i'} \bigcup \{j\} \succ_{j'} S_{i'}$, $j \neq j'$, $\forall j' \in S_{i'}$.
\end{itemize}

Note that when a coalitional structure is Nash-stable, it is also individually-stable~\cite{anna2002stability} since Nash-stability is a subset of individual-stability.

\begin{proposition}
The final partition $\Pi^*=\{S_1^*,\ldots, S_i^*,\ldots, S_I^*\}$ is a Nash-stable and individually-stable coalitional structure.
\end{proposition}

\begin{proof}
Given any current coalitional structure $\Pi_{curr}=\{S_1,\ldots,S_i,\ldots, S_I\}$, where $S_i \cap S_{i'}=\emptyset$ for $i \neq i'$, $\bigcup_{i=1}^IS_i=\mathcal{J}$, switch operations are performed for any worker $j \in \mathcal{J}$ to either leave or join a coalition. The current partition $\Pi_{curr}$ is updated when the utility of any worker $j \in \mathcal{J}$ in any coalition $S_i$ is higher by leaving its current coalition $S_i$ and joining another coalition $S_{i'} \in \Pi_{curr}$, for $i\neq i'$. According to Equation~(\ref{eqn:function}) which defines the preference function of each worker, the worker does not visit the coalitions that are contained in its history set. Therefore, the hedonic coalitional formation algorithm generates a sequence of coalitional structures where each coalitional structure has not been visited before. The algorithm will eventually terminate at a final coalitional structure $\Pi^*=\{S_1^*,\ldots, S_i^*,\ldots,S_I^*\}$, where there is no incentive for any worker to change its current coalition. In other words, the utility of each worker $j$, $\forall j \in \mathcal{J}$, is maximized given the final coalitional structure $\Pi^*$.

Suppose we assume that the final coalitional structure $\Pi^*$ is not Nash-stable. This implies that there is a worker $j \in \mathcal{J}$ that has incentive to change its current coalition. As a result, by leaving its current coalition and joining another coalition, the coalitional structure $\Pi^*$ is updated based on the switch rule defined in Definition~\ref{def:switch}. Thus, the coalitional structure is not final, which does not align with our assumption that the final coalitional structure is not Nash-stable. Therefore, the final coalitional structure $\Pi^*$ must be Nash-stable. Since final coalitional structure $\Pi^*$ is Nash-stable, it is also individually-stable.
\end{proof}

\subsection{Hedonic Coalition Formation Algorithm}

The algorithm for the hedonic coalition formation is presented in Algorithm~\ref{algo}. 

The hedonic coalition formation game is based on the switch rule defined in Definition~\ref{def:switch}. The switch operation is illustrated from the perspective of worker $j \in \mathcal{J}$. The worker $j$ decides to leave its current coalition $S_i$ and join another coalition $S_{i'}$ where $i\neq i'$ and $S_{i'} \subseteq \Pi_{curr}$ if and only if the worker $j$ achieves higher utility by joining coalition $S_{i'}$ than that of the current coalition $S_i$. The worker $j$ first compute its utility in the current coalition $S_i$, $x^{S_i}_j$ (line~8). Given the current coalitional structure $\Pi_{curr}$, the worker $j$ evaluates the other coalitions $S_{i'}$, for $i\neq i'$ that it could possibly join (line~9-10). Specifically, the worker $j$ computes the utility that it achieves if it joins another coalition $S_{i'}$, $x^{S_{i'}}_j$ (line~10). If the utility of worker $j$ for joining coalition $S_{i'}$ is higher than that of the current coalition $S_i$ and the coalition $S_{i'}$ is not found in the history set of worker $j$, $h(j)$, the worker $j$ performs switch operation (line~11-16). In particular, the worker $j$ first updates its history set by adding the current coalition $S_i$ into $h(j)$ (line~12). The worker $j$ leaves its current coalition $S_i$ (line~13) and joins the new coalition $S_{i'}$ (line~14). Then, given the new coalition $S_{i'}$, the current coalition and coalitional structure are updated (line~15-16). On the other hand, if the utility of worker for joining coalition $S_{i'}$ is lower than that of the current coalition $S_i$ or the new coalition $S_{i'}$ has been visited before, the worker $j$ does not leave its current coalition $S_i$, thus there is no change in the coalitional structure. For the next iteration, the worker $j$ will consider to join other possible coalition $S_{i'} \in \Pi_{curr}$, for $i\neq i'$. The process is repeated for all workers $j \in \mathcal{J}$. The switch mechanism terminates when there is no more change to the current coalitional structure, $\Pi_{curr}$. In other words, there is no worker $j$ ($\forall j \in \mathcal{J}$) that is able to achieve higher utility by leaving its current coalition and join any other coalition in the current coalition structure, $\Pi_{curr}$. At the end of the switch mechanism, the algorithm returns the final coalitional structure $\Pi^*=\{S_1^*,\ldots,S_i^*,\ldots,S_I^*\}$ that is Nash-stable (line~20). Consequently, the total amount of CPU power that each coalition can be computed, in which the competition between the different cluster heads are discussed in the next section. 

%\begin{proposition}
%The formation of a grand coalition, $\mathcal{J}$ where all workers $j \in \mathcal{J}$ forms a single coalition to support any cluster head $i \in \mathcal{I}$, is not stable.
%\end{proposition}
%
%\begin{proof}
%
%\end{proof}

\begin{algorithm}[t]
\caption{Algorithm for Hedonic Coalition formation of Workers using Switch Rule.}
%\footnotesize
\label{algo}
\begin{algorithmic}[1]
 \renewcommand{\algorithmicrequire}{\textbf{Input:}}
 \renewcommand{\algorithmicensure}{\textbf{Output:}}
 \REQUIRE Set of workers, $\mathcal{J}=\{1,\ldots, j,\ldots, J\}$, set of cluster heads, $\mathcal{I}=\{1, \ldots,i,\ldots,I\}$
 \ENSURE Final coalitional structure $\Pi^*=\{S_1^*,\ldots,S_i^*,\ldots,S_I^*\}$
 
 \STATE $\Pi^*=\emptyset$
 \STATE Initialize history set for all workers, i.e., $h(j)=\emptyset$, $\forall j \in \mathcal{J}$
 \STATE Given $J$ workers, initialize a coalitional structure $\Pi_{curr}$ where workers are randomly allocated to the $I$ coalitions
 
 \STATE \textbf{\emph{\underline{Switch Rule:}}}
 
 \WHILE {$\Pi_{curr}\neq \Pi^*$}
 	\STATE Update the final coalitional structure such that $\Pi^*=\Pi_{curr}$
	\FOR {\textbf{each} worker $j \in \mathcal{J}$ (worker $j$ is in coalition $S_i \in$ coalitional structure $\Pi_{curr}$)}
		\STATE Compute $x^{S_i}_j$
		%\STATE Compute the utility of worker $j$ in coalition $S_i$, $x^{S_i}_j$
		%\STATE Initialize the maximum utility of worker $j$ such that $x^{max}_j \gets x^{S_i}_j$
		%\STATE Consider the possible coalitions $S_{i'} \in \Pi_{curr}$, $i\neq i'$
		\FOR {\textbf{each} possible coalition $S_{i'} \in \Pi_{curr}$, $i \neq i'$ }
			\STATE Compute $x^{S_{i'}}_j$
			%\STATE Compute the utility of joining a new coalition $S_{i'}$, $x^{S_{i'}}_j$
		\IF {$x^{S_{i'}}_j > x^{S_i}_j$ \AND $S_{i'} \notin h(j)$ }
			\STATE Worker $j$ updates its history set, $h(j)$ by adding the current coalition $S_i$ into $h(j)$
			\STATE Worker $j$ leaves its current coalition, $S_i=S_i \backslash \{j\}$ 
			\STATE Worker $j$ joins the new coalition that increases its utility, $S_{i'}=S_{i'}\bigcup \{j\}$
			\STATE Update current coalition of worker $j$, $S_i \gets S_{i'}$
		\STATE Update current coalitional structure $\Pi_{curr} \gets \Pi_{curr} \backslash \{S_i,S_{i'}\} \bigcup \{S_{i}\backslash \{j\}, S_{i'}\bigcup\{j\}\}$
			%\STATE Update the maximum utility of worker $j$, $x^{max}_j \gets x^{S_{i'}}_j$
		\ENDIF
		\ENDFOR	
	\ENDFOR	
\ENDWHILE

\RETURN Final coalitional structure $\Pi^*=\{S_1^*,\ldots,S_i^*,\ldots,S_I^*\}$ that is Nash-stable
\end{algorithmic}
\end{algorithm}

\section{Upper-level All-pay Auction}
\label{sec:upper}

Since the cluster heads, given their coalitions of workers formed, may have several computation tasks to complete, they only allocate a fraction of their CPU power to the CDC tasks. Hence, in order to incentivize the cluster heads to allocate more CPU power for the allocated CDC subtasks, we present the design of an all-pay auction in this section. In this all-pay auction, the master is the auctioneer whereas the cluster heads are the bidders. The bid of a cluster head is represented by the CPU power that it allocates for its CDC subtask, which, for example, can be measured from the computation latency incurred for the CDC subtask. All $I$ cluster heads, i.e., bidders, pay their bids regardless of whether they win or lose the auction. We first discuss the utilities of both the master and the cluster heads. Then, we present the design of an all-pay auction.

\subsection{Utility of the Master}

%The master aims to complete the computation task by leveraging on the computational capabilities of the workers. In order to incentivize the workers to allocate more CPU power to complete the computations on the allocated submatrices, the master offers rewards to the workers. 
Given that the master only needs the computed results from $K$ cluster heads to reconstruct the final result, the master offers $K$ rewards, represented by the set $\mathcal{K}=\{1,\ldots, k,\ldots, K\}$ where $K \leq I$. Specifically, there are $K$ rewards for which $I$ cluster heads compete. The effect of different reward structures is discussed later in details in Section~\ref{subsec:reward}. Since only $K$ rewards are offered, $I-K$ cluster heads do not receive any reward from the master, even though they perform the matrix multiplication computations given the allocated submatrices. The all-pay auction is designed such that the cluster heads are incentivized to allocate their CPU power, even if there is a possibility that they may not win any reward. 

The size of reward $k$ is represented by $M_k$. The cluster head that allocates larger CPU power is offered larger reward. In particular, the cluster head that allocates the largest amount of CPU power receives a reward of $M_1$, the cluster head with the second largest allocation receives reward $M_2$ and the cluster head with the $k$-th largest allocation of CPU power is offered reward $M_k$. If two or more cluster heads allocate the same amount of CPU power to perform the CDC tasks, ties will be randomly broken. In other words, if both cluster heads are ranked $k$, one is ranked $k$ and the other is ranked $k+1$. Hence, without loss of generality, $M_1\geq M_2 \geq \cdots \geq M_K > 0$. The total amount of reward offered by the master is denoted by $\sigma$, i.e., $\sigma=\sum_{k=1}^{K}M_{k}$. The master broadcasts the information of size of total reward and the structure of rewards to the workers. The aim of the master is to share the entire fixed reward to maximize the CPU power allocated by the cluster heads. 

As such, the expected utility of the master, $\pi$ is expressed as follows:
\begin{equation}
\label{profit}
\pi=\mathbb{E}[\phi(\tau_{1:I}+\tau_{2:I}+\cdots+\tau_{K:I})-\sigma],
\end{equation}
where $\phi$ is the unit worth of CPU power to the master and $\tau_{k:I}$ represents the order statistics of the cluster head's CPU power allocation. Specifically, $\tau_{1:I}$ and $\tau_{k:I}$ denote the highest and $k$-th highest CPU power allocation respectively among $I$ cluster heads.

\subsection{Utility of the Cluster Head}

To perform the local computations on the allocated CDC subtask, each cluster head $i$ consumes computational energy, $e_i$, which is defined as:
\begin{equation}
e_i=\kappa a_{i}(\tau_i)^2,
\end{equation} 
where $\kappa$ is the effective switch coefficient that depends on the chip architecture \cite{zhang2018energy}, $a_{i}$ is the total number of CPU cycles required to complete the allocated computation subtask and $\tau_i$ is the CPU power allocated by cluster head $i$ for the CDC subtask. In other words, $v(S_i)-\tau_i$ is the amount of CPU power allocated by cluster head $i$ for other computation tasks. By using the polynomial codes, the computation task is evenly partitioned and distributed among all cluster heads. As a result, the total number of CPU cycles that are needed to complete the allocated computation tasks is the same for all cluster heads, i.e., $a_{i}=a, \; \forall i \in \mathcal{I}$. The unit cost of computational energy incurred by cluster head $i$, $\forall i \in\mathcal{I}$, is denoted by $\theta^p$, where the unit cost of computational energy is the same for all cluster heads. Besides, each cluster head $i$ also requires communication energy $c_i$ to communicate with the master. Similarly, the unit cost of communication energy is the same for all cluster head where the unit cost of communication energy incurred by cluster head $i$, $\forall i \in \mathcal{I}$, is denoted by $\theta^c$.

%Each worker $i$ is characterized by $\theta_i$, where $\theta_i \in [\underline{\theta},\bar{\theta}]$ and $\underline{\theta}>0$. $\theta_i$ represents the marginal cost of each unit of computational energy incurred by worker $i$. $\bar{\theta}$ and $\underline{\theta}$ represent the highest and lowest marginal cost of computational energy respectively. 

Each cluster head $i$ has a valuation $v_{i}$ for the total reward $\sigma$. For example, in practical scenarios, the valuations for the total reward can be determined by how much the cluster heads can benefit from the reward, which is a user preference parameter. In particular, the cluster heads value the reward more if they need the reward for some important purposes, e.g. upgrading of their hardware components. The cluster heads' valuations, $v_i$,  $\forall i \in \mathcal{I}$, are independently drawn from $v_i \in [\underline{v},\bar{v}]$ such that $\underline{v}$ and $\bar{v}$ are strictly positive given $F(v)$, where $F(v)$ is the cumulative distribution function (CDF) of $v$. The total cost of cluster head $i$ is represented by $\theta^{p}e_{i}+\theta^{c}c_{i}$. As a result, the utility of cluster head $i$ for winning reward $M_{k}$, $\forall k \in \mathcal{K}$, is expressed as:
\begin{equation}
\label{eqn:payoff}
\alpha_{i}=
	\begin{cases}
   	v_{i}M_{k}-\theta^{p}e_{i}-\theta^{c}c_{i}, & \text{if cluster head $i$ wins $M_{k}$}, \\
   	-\theta e_{i}, & \text{otherwise}.
 	\end{cases}
\end{equation}
%Note that given the same total energy cost incurred, the worker with a higher valuation for reward $M_k$ is better off than that of the worker with a lower valuation for reward $M_{k}$.

\subsection{Design of an All-pay Auction}

Each cluster head $i$ knows its own valuation, $v_{i}$ but does not know the valuation of any other cluster head, $i'\neq i$. This establishes a one-dimensional incomplete information setting. In addition, if each cluster head has different unit costs of computational and communication energy which are only known to itself, we consider the three-dimensional incomplete information setting. The dimension of private information can be reduced following the procedure in~\cite{yoon2012optimal}. In this work, we consider a one-dimensional incomplete information setting where the unit costs of computational and communication energy are the same for all cluster heads but the cluster heads' valuations are heterogeneous and private. 

Given the utility of cluster head $i$, $\alpha_i$ in Equation~(\ref{eqn:payoff}), the objective of cluster head $i$ to maximize its expected utility, $u_{i}$, is defined as follows:
\begin{equation}
\label{eqn:utility}
%\max_{z_i}u_{i}&=v_{i}\sum_{k=1}^Kp_{i}^kM_{k}-\theta_{i}e_{i},\\
\max_{\tau_i}u_{i}=v_{i}\sum_{k=1}^Kp_{i}^k{M_{k}}-\theta^{p}\kappa a(\tau_{i})^2-\theta^{c}c_i,
\end{equation}
where $p_{i}^k$ is the winning probability of reward $M_k$ by cluster head~$i$.

Although the cluster head does not know exactly the valuations of other cluster heads, it knows the distribution of the other cluster heads' valuations based on past interactions, which is a common knowledge to all cluster heads and the master. In our model, we consider that the valuations of all cluster heads are drawn from the same distribution, which constitutes a symmetric Bayesian game where the prior is the distribution of the cluster heads' valuations.

\begin{definition} \cite{tie2014optimal} A pure-strategy Bayesian Nash equilibrium is a strategy profile $\mathbf{\tau}^*=(\tau_1^*,\ldots,\tau_i^*,\ldots,\tau_I^*)$ that satisfies
$$u_{i}(\tau_{i}^*,\mathbf{\tau}_{-i}^*) \geq u_{i}(\tau_{i},\mathbf{\tau}_{-i}^*), \; \forall i \in \mathcal{I}. $$
\end{definition}
The subscript $-i$ represents the index of other cluster heads other than cluster head $i$. Specifically, $\mathbf{\tau}_{-i}^*=(\tau_1^*,\tau_2^*,\ldots,\tau_{i-1}^*,\tau_{i+1}^*,\ldots,\tau_I^*)$ represents the equilibrium CPU power allocations of all other cluster heads other than CPU power allocation of cluster head $i$. At the Bayesian Nash equilibrium, given the belief of cluster head $i$, $\forall i \in \mathcal{I}$ about the valuations and that the CPU power allocated by other cluster heads, $i'$ where $i \neq i'$ are at equilibrium, $\mathbf{\tau}^*_{i'}$, cluster head $i$ aims to maximize its expected utility. 

\begin{proposition}Under incomplete information setting, the all-pay auction admits a pure-strategy Bayesian Nash equilibrium that is strictly monotonic where the bid of a cluster head strictly increases in its valuation.
\end{proposition}

%\begin{proof}
%The proof is omitted due to space constraints.
%\end{proof}

Since the equilibrium CPU power allocation of cluster head~$i$, which is represented by $\tau_{i}^*$, is a strictly monotonically increasing function of its valuation $v_{i}^*$, we express the equilibrium strategy of cluster head $i$ as a function represented by $\beta(\cdot)$, i.e., $\tau_{i}^*=\beta_{i}(v_{i})$. Given the strict monotonicity, the inverse function also exists where $v_i(\cdot)=\beta_i^{-1}(\cdot)$ and it is an increasing function. Due to the incomplete information setting, the objective of cluster head $i$ to maximize its expected utility in Equation~(\ref{eqn:utility}) can be expressed as follows:
\begin{equation}
\max_{\tau_i}u_{i}=v_{i}\sum_{k=1}^Kp_{i}^k(\tau_{i},\beta_{i'}(v_{i'})){M_{k}}-c(\beta(v_{i})),
\end{equation}
where the cost of cluster head $i$ is represented by the function $c(\cdot)=\theta^{p}\kappa a(\beta(v_{i}))^2+\theta^{c}c_i$.

Since the cluster heads are symmetric, i.e., the valuations of cluster heads are drawn from the same distribution, the symmetric equilibrium strategy for each cluster head $i$, $\forall i \in \mathcal{I}$ can be derived. We first assume that there are $I$ rewards, where $M_{1}\geq M_{2}\geq \cdots \geq M_{K}>M_{K+1}=M_{K+2}=\cdots=M_{I}=0$. The valuations of the cluster heads, $v_1,\ldots, v_{i},\ldots, v_{I}$ are ranked and represented by its order statistics, which are expressed as $v_{1:I}\geq v_{2:I} \geq\cdots\geq v_{I:I}$. In particular, $v_{k:I}$ represents the $k$-th highest valuation among the $I$ valuations which are drawn from a common distribution $F(v)$. Given the order statistics of the cluster heads' valuations, $\forall i \in \mathcal{I}$, the corresponding cumulative distribution function and probability density function are represented by $F_{k:I}$ and $f_{k:I}$ respectively. Specifically, the cumulative distribution function $F_{k:I}(v)$ for the $k$-th order statistics in sample of size $I$ is expressed as follows:

\begin{equation}
F_{k:I}(v)=\sum_{r=0}^{k-1}F(v)^{I-r}[1-F(v)]^r.
\end{equation}

The corresponding probability density function $f_{k:I}(v)$ for $k$-th order statistics in sample of size $I$ is expressed as follows:

\begin{equation}
f_{k:I}(v)=\frac{I!}{(k-1)!(I-k)!}F(v)^{(I-k)}[1-F(v)]^{k-1}f(v).
\end{equation}

Similarly, when dealing with the valuations of all cluster heads, other than that of cluster head $i$, the order statistic is represented by $v_{k:I-1}$, which represents the $k$-th highest valuation among the $I-1$ valuations. The corresponding cumulative distribution function and probability density function are represented by $F_{k:I-1}$ and $f_{k:I-1}$ respectively.

Given that other cluster heads $i'$, where $i' \neq i$, follow a symmetric, increasing and differentiable equilibrium strategy $\beta(\cdot)$, cluster head $i$ will never choose to allocate a CPU power greater than the equilibrium strategy given the highest valuation. In other words, cluster head $i$ will never allocate $\tau_i>\beta(\bar{v})$. Besides, the optimal strategy of the cluster head with lowest valuation $\underline{v}$ is not to allocate any CPU power. On one hand, when the number of rewards offered is smaller than the number of cluster heads, i.e., $K<I$, the cluster head with lowest valuation $\underline{v}$ will not win any reward. On the other hand, when the number of rewards offered is larger than or equal the number of cluster heads, i.e., $K\geq I$, the cluster head with lowest valuation $\underline{v}$ will win a reward without allocating any CPU power. Hence, $u_i({\underline{v}})=0$. With this, the expected utility of cluster head $i$ with valuation $v_{i}$ and CPU power allocation $\tau_i=\beta(v_i)$ is expressed as follows:
\begin{equation}
\label{eqn:utilityprob}
u_{i}=v_{i}\sum_{k=1}^I[F_{k:I-1}(v_i)-F_{k-1:I-1}(v_i)]{M_{k}}-c(\beta(v_i)),
\end{equation}
since $M_{k+1}=\cdots=M_{I-1}=M_{I}=0$, $F_{0:I-1}(\tau_i)\equiv0$ and $F_{I:I-1}(\tau_i)\equiv1$. %The term $w_{i}$ is used instead of $v_i$ to avoid confusion of terms in the integration expression that is shown later.

By differentiating Equation~(\ref{eqn:utilityprob}) with respect to the variable $w_i$ and equating the result to zero, we obtain the following:
\begin{equation}
\label{differentiate}
0=v_{i}\sum_{k=1}^I[f_{k:I-1}(v_i)-f_{k-1:I-1}(v_i)]{M_{k}}-c'(\beta(v_i))\beta'(v_i).
\end{equation}

When maximized, the marginal value of the reward is equivalent to the marginal cost of the CPU power. Since we have the differentiated function $c'(\cdot)$, the function $c(\cdot)$ can be found by using the integral of Equation~(\ref{differentiate}). At equilibrium, when the expected utility of cluster head $i$, $\forall i \in \mathcal{I}$, is maximized, we have the following:
\begin{equation}
\begin{split}
c(\beta(v_i))&=\sum_{k=1}^{I}M_{k}\int_{\underline{v}}^{v_i}v_{i}[f_{k:I-1}(v_{i})-f_{k-1:I-1}(v_{i})]dv_{i}\\
&=\sum_{k=1}^{I-1}(M_{k}-M_{k+1})\int_{\underline{v}}^{v_i}v_{i}f_{k:I-1}(v_{i})dv_{i}.
\end{split}
\end{equation}

Thus the equilibrium strategy for cluster head $i$ with valuation $v_{i}$, $\forall i \in \mathcal{I}$, is expressed as:
\begin{equation}
\tau_{i}^*=\beta(v_{i})=c^{-1}\left(\sum_{k=1}^{I-1}(M_{k}-M_{k+1})\int_{\underline{v}}^{v_i}v_{i}f_{k:I-1}(v_{i})dv_{i} \right).
\end{equation}

Given the equilibrium strategy of cluster head $i$, $\forall i \in \mathcal{I}$, the master aims to maximize its expected utility, $\pi$. By using the polynomial codes, the master is able to reconstruct the final result by using the computed results from $K$ cluster heads. Since the master shares the fixed reward $\sigma$ completely, the maximization problem in Equation~(\ref{profit}) is equivalent to maximizing the allocation of CPU power, which is expressed as follows: 
\begin{multline}
\pi=\mathbb{E}[\beta(v_{1:I})+\beta(v_{2:I})+\cdots+\beta(v_{K:I})]\\
\shoveleft{=\sum_{i=1}^K\int_{\underline{v}}^{\bar{v}}\beta(v)dF_{i:I}(v)}\\
\shoveleft{=K\int_{\underline{v}}^{\bar{v}}\beta(v)dF(v)}\\
\shoveleft{=K\int_{\underline{v}}^{\bar{v}}c^{-1}\left(\sum_{k=1}^{I-1}(M_{k}-M_{k+1})\int_{\underline{v}}^{v}vf_{k:I-1}(v)dv \right)dF(v)}\\
\shoveleft{=K\int_{\underline{v}}^{\bar{v}}c^{-1}(\sum_{k=1}^{I}M_{k}\int_{\underline{v}}^{v}v[f_{k:I-1}(v)}\\
{ -f_{k-1:I-1}(v)dv])dF(v)}.
\end{multline}

Since the equilibrium strategy of cluster head $i$, $\forall i \in \mathcal{I}$, is affected by the reward structure, the master needs to determine the structure of the rewards such that it maximizes the CPU power allocation of the cluster heads, thereby maximizing its own utility, $\pi$.

\subsection{Reward Structure}
\label{subsec:reward}

Given that the master shares the total amount of the reward,~$\sigma$, the design of the optimal reward sequence is important to maximize the CPU power allocation of the cluster heads since the equilibrium strategies of the cluster heads depend on the differences between consecutive rewards.

The master needs to first decide whether to allocate the total amount of reward, $\sigma$ to only one winner, i.e., winner-take-all reward structure, or to split the reward into several smaller rewards.

\begin{proposition}

Given that the cost functions are convex, it is not optimal to offer only one reward where $M_1=\sigma$ and $ M_2=\cdots=M_K=\cdots=M_I=0$ since $\frac{\partial \pi}{\partial M_{k-1}}-\frac{\partial \pi}{\partial M_{k}} <0$, for $k=2,\ldots,I$. In particular, if $\frac{\partial \pi}{\partial M_{1}}-\frac{\partial \pi}{\partial M_{2}} <0$, it is not optimal to offer only a reward. 
\end{proposition}

\begin{proof}

Following the procedure in~\cite{yoon2012optimal}, we show that it is not optimal to offer a single reward given the cost functions of the cluster heads are convex. %The details of the proof are omitted due to space constraints.
\begin{multline*}
\frac{\partial \pi}{\partial M_{1}}-\frac{\partial \pi}{\partial M_{2}}=K\int_{\underline{v}}^{\bar{v}}(c^{-1})'(\int_{\underline{v}}^{v}vf_{1:I-1}(v)dv)\\
\times \left\{2\int_{\underline{v}}^{v}vf_{1:I-1}(v)dv-\int_{\underline{v}}^{v}vf_{2:I-1}(v)dv\right \}dF(v).
\end{multline*}

\begin{multline*}
\frac{\partial}{\partial v}\left( 2\int_{\underline{v}}^{v}vf_{1:I-1}(v)dv-\int_{\underline{v}}^{v}vf_{2:I-1}(v)dv\right )\\
\shoveleft{=v \{2\frac{(I-1)!}{(I-2)!}F(v)^{I-2}f(v) -\frac{(I-1)!}{(I-3)!}F(v)^{I-3}}\\
\shoveright{\times[1-F(v)]f(v)\}}\\
\shoveleft{=vf(v)F(v)^{I-3}(2(I-1)F(v)-(I-1)(I-2)[1-F(v)])}.
\end{multline*}

Let $x=F(v)$, the expression above is simplified to:
\begin{multline*}
\frac{\partial}{\partial v}\left( 2\int_{\underline{v}}^{v}vf_{1:I-1}(v)dv-\int_{\underline{v}}^{v}vf_{2:I-1}(v)dv\right )\\
=vf(v)x^{I-3}(2(I-1)x-(I-1)(I-2)[1-x]).
\end{multline*}

When $x=0$, $\frac{\partial}{\partial v}(\cdot)=vf(v)x^{I-3}[-(I-1)(I-2)]<0$. When $x=1$, $\frac{\partial}{\partial v}(\cdot)=vf(v)x^{I-3}2(I-1)>0$. As a result, there is $\hat{x}=F(\hat{v})$ with $\hat{v}\in (\underline{v},\bar{v})$ such that
\begin{equation*}
 \frac{\partial}{\partial v}\left( 2\int_{\underline{v}}^{v}vf_{1:I-1}(v)dv-\int_{\underline{v}}^{v}vf_{2:I-1}(v)dv\right )>0,
\end{equation*}
if and only if $v>\hat{v}$. As such, this implies that there is $v^*\in (\underline{v},\bar{v})$ such that
\begin{equation*}
2\int_{\underline{v}}^{v}vf_{1:I-1}(v)dv-\int_{\underline{v}}^{v}vf_{2:I-1}(v)dv>0,
\end{equation*}
if and only if $v>\hat{v^*}$.
Given that 
\begin{equation*}
\int_{\underline{v}}^{\bar{v}}2\int_{\underline{v}}^{v}vf_{1:I-1}(v)dv-\int_{\underline{v}}^{v}vf_{2:I-1}(v)dvdF(v)>0,
\end{equation*}
and $(c^{-1})'(\cdot)<0$ and $(c^{-1})''(\cdot)\geq0$ due to the convexity of the cost function, $\frac{\partial \pi}{\partial M_{1}}-\frac{\partial \pi}{\partial M_{2}}<0$. Hence, it is not optimal to allocate only one reward to the cluster head which allocates the largest amount of CPU power, where $M_1=\sigma$ and $M_2=\cdots=M_I=0$. Note that the similar procedure can be used to proof for the general case of $\frac{\partial \pi}{\partial M_{1}}-\frac{\partial \pi}{\partial M_{k}}<0$ for $k=2,\ldots, I-1$.
\end{proof} 

Since the winner-take-all reward structure is not optimal, the master is better off offering multiple rewards. Given that $K$ rewards are offered, the master can consider several reward sequences such as (i) homogeneous reward sequence, (ii) arithmetic reward sequence and (iii) geometric reward sequence. Specifically, the reward sequence is expressed as follows:
\begin{itemize}
\item Homogeneous reward sequence: $M_{k}=M_{k+1}$,
\item Arithmetic reward sequence: $M_{k}-M_{k+1}=\gamma$, $\gamma>0$,
\item Geometric reward sequence: $M_{k+1}=\eta M_{k}$, $0\leq\eta\leq 1$,
\end{itemize}
where $\gamma$ and $\eta$ are constants.

%In the next section, we show the effects of different reward structures on the allocation of CPU power by the cluster heads.

\section{Simulation Results}
\label{sec:simulate}

In this section, we evaluate the two-level game theoretic approach. We first analyze the hedonic coalition formation game that maximizes the utilities of the workers, followed by the all-pay auction. In particular, we evaluate the behaviour of the cluster heads in allocating their CPU power for the CDC subtasks. Table~\ref{tab:simulation} summarizes the simulation parameter values. 

We consider a nomalized total amount of reward $\sigma$ of $1$, i.e., $\sigma=\sum_{k=1}^{K}M_{k}=1$. We also set $m=n=2$ (see ``Task Allocation'' step in Section~\ref{sec:system}) and assume that the cluster heads are able to store equal size of the input matrices, $\mathbf{A}$ and $\mathbf{B}$ such that $m=n=2$.

\begin{table}
\caption{System Simulation Parameter Values.} 
\label{tab:simulation}
\centering
\renewcommand\arraystretch{1.5}
\begin{tabularx}{8.7cm}{XlX}
%\begin{tabular}{p{5.5cm} | p{2.2cm} }
\hline \hline
\textbf{Parameter}& \textbf{Values}\\ [0.5ex]
\hline 
CPU power of worker $j$, $z_j$ & $[100,450]$\\
CPU power of cluster head $i$, $z_i$ & $[750,1750]$\\
Communication cost between worker $j$ and cluster head $i$, $\mu_{ij}$ & 2\\
Unit cost of computational energy, $\theta^{p}$ & 1\\
Unit cost of communication energy, $\theta^{c}$ & 1\\
Communication energy for required by cluster head $i$, $c_i$ & $5$\\
Effective switch coefficient, $\kappa$ \cite{hao2018energy} & $10^{-25}$\\
Total number of CPU cycles required, $a$ & $5\times 10^{9}$\\
Valuation of cluster head $i$, $v_{i}$ & $\sim U[0,1]$\\

\hline 
\end{tabularx}
\end{table}

\subsection{Lower-level Hedonic Coalition Formation}

\begin{table}[!htb]
  
    \begin{minipage}{.45\linewidth}
      \caption{\footnotesize{Simulation Parameter Values of the Cluster Heads.} }
      \label{clusterhead}
      \centering
      \renewcommand\arraystretch{1.5}
        \begin{tabularx}{\columnwidth}{>{\centering \arraybackslash \hsize=0.5\hsize}X >{\centering \arraybackslash \hsize=0.4\hsize}X >{\centering \arraybackslash \hsize=0.4\hsize}X} 
        \hline\hline
        Cluster Head (CH) ID& CPU Power (W) & Reward\\ [0.5ex]
        \hline
         CH 1 & 750 & 100\\ 
	CH 2 & 1000 & 90 \\
	CH 3 & 1250 & 80 \\ 
	CH 4 & 1500 & 70 \\ 
	CH 5 & 1750 & 60 \\
	\hline 
        \end{tabularx}
    \end{minipage}%
    \hspace{2em}
    \begin{minipage}{.45\linewidth}
      \centering
        \caption{\footnotesize{Simulation Parameter Values of the Workers.}} 
        \label{accesspoint}
        \renewcommand\arraystretch{1.5}
        \begin{tabularx}{\columnwidth}{>{\centering \arraybackslash \hsize=0.6\hsize}X >{\centering \arraybackslash\hsize=0.4\hsize}X >{\centering \arraybackslash \hsize=0.35\hsize}X} 
            	\hline\hline
	Worker ID& CPU Power (W) & Unit Cost\\ [0.5ex]
	\hline 
	Worker 1 & 100 & 0.01\\ 
	Worker 2 & 150 & 0.02\\
	Worker 3 & 200 & 0.03\\ 
	Worker 4 & 250 & 0.04\\ 
	Worker 5 & 300 & 0.05\\
	Worker 6 & 350 & 0.06\\ 
	Worker 7 & 400 & 0.07\\ 
	Worker 8 & 450 & 0.08\\ 
	\hline 
        \end{tabularx}
    \end{minipage} 
\end{table}

In the network, there are $5$ cluster heads and $8$ workers with different CPU powers. We consider the hedonic coalition formation game among the cluster heads and workers. The objective of each worker is to maximize its own utility, which depends solely on the members of the coalition it belongs to. In particular, the utility of each worker is affected by its proportion of CPU power in the coalition. The simulation parameter values of the cluster heads and the workers are listed in Tables~\ref{clusterhead} and~\ref{accesspoint} respectively. 

The hedonic coalition formation game allows the workers to decide which cluster head to join. In order to decide whether to stay in or leave a coalition, the workers adopt the switch rule. Figure~\ref{fig:hedonicformation} illustrates the mechanism of the switch operations. Initially, the workers are randomly assigned to the cluster heads. Each time the coalitional structure changes, each worker evaluates its utility by comparing the utility achieved in the current coalition against the utility gain from joining other possible coalitions. As a result, each worker may perform more than one switch operation. As an example, worker~1 achieves a utility of $89$ by joining cluster head~2. As workers~3 and 5 join the coalition in supporting cluster head~2, the utility of worker~1 decreases to $14$. Worker~1 then decides to leave cluster head~2 and joins worker~7 to support cluster head~1 as it gains a higher utility of $19$. However, when worker~2 joins the coalition to support cluster head~1, worker~1's utility decreases to $14.4$. As such, worker~1 decides to perform a switch operation again where it joins worker~6 in supporting cluster head~3, achieving a utility of $16.8$.

\begin{figure}
\includegraphics[width=\linewidth]{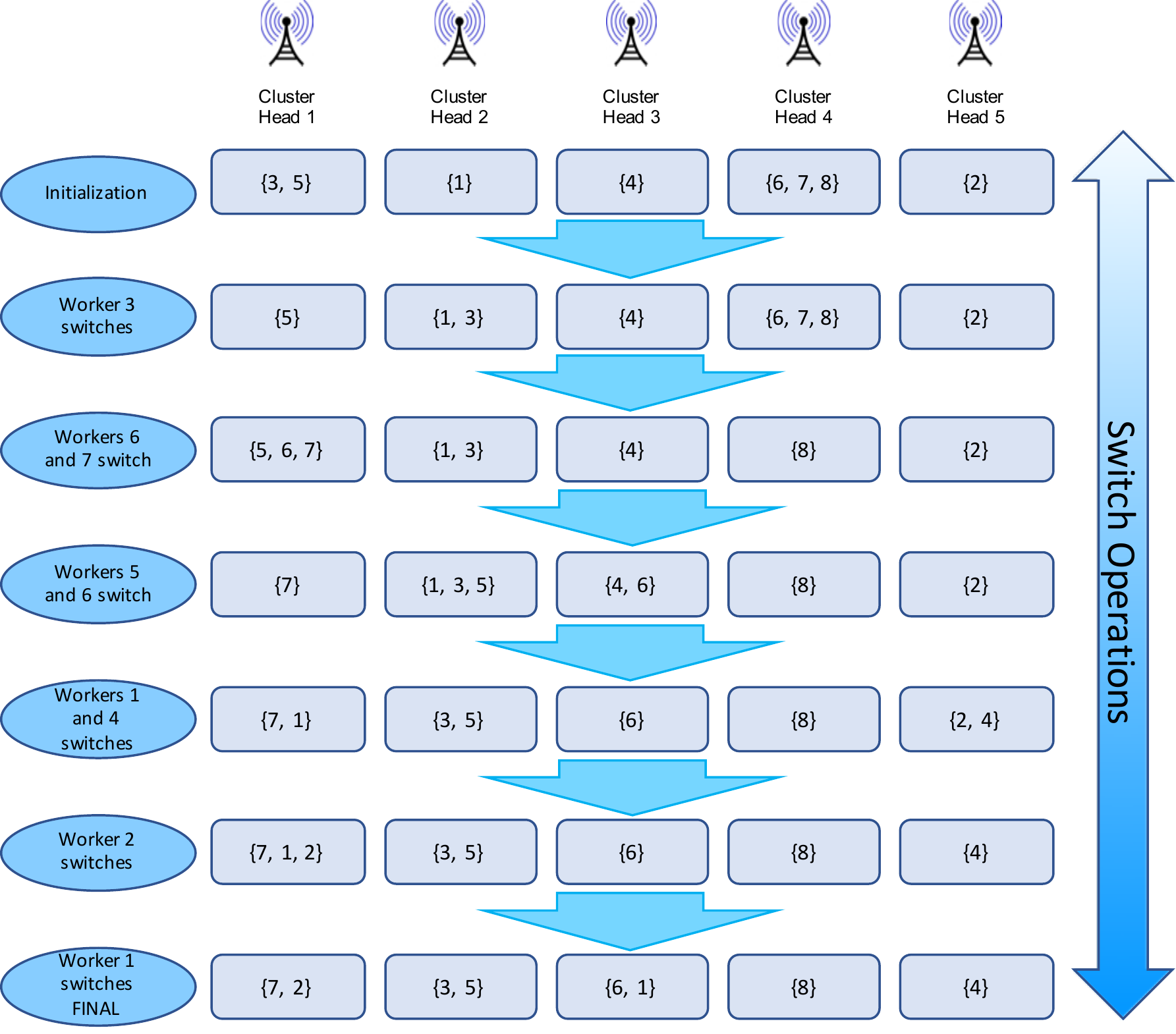}
\caption{Switch operations of the hedonic coalition formation game.}
\label{fig:hedonicformation}
\end{figure}

\begin{figure}
\centering
\includegraphics[width=\linewidth]{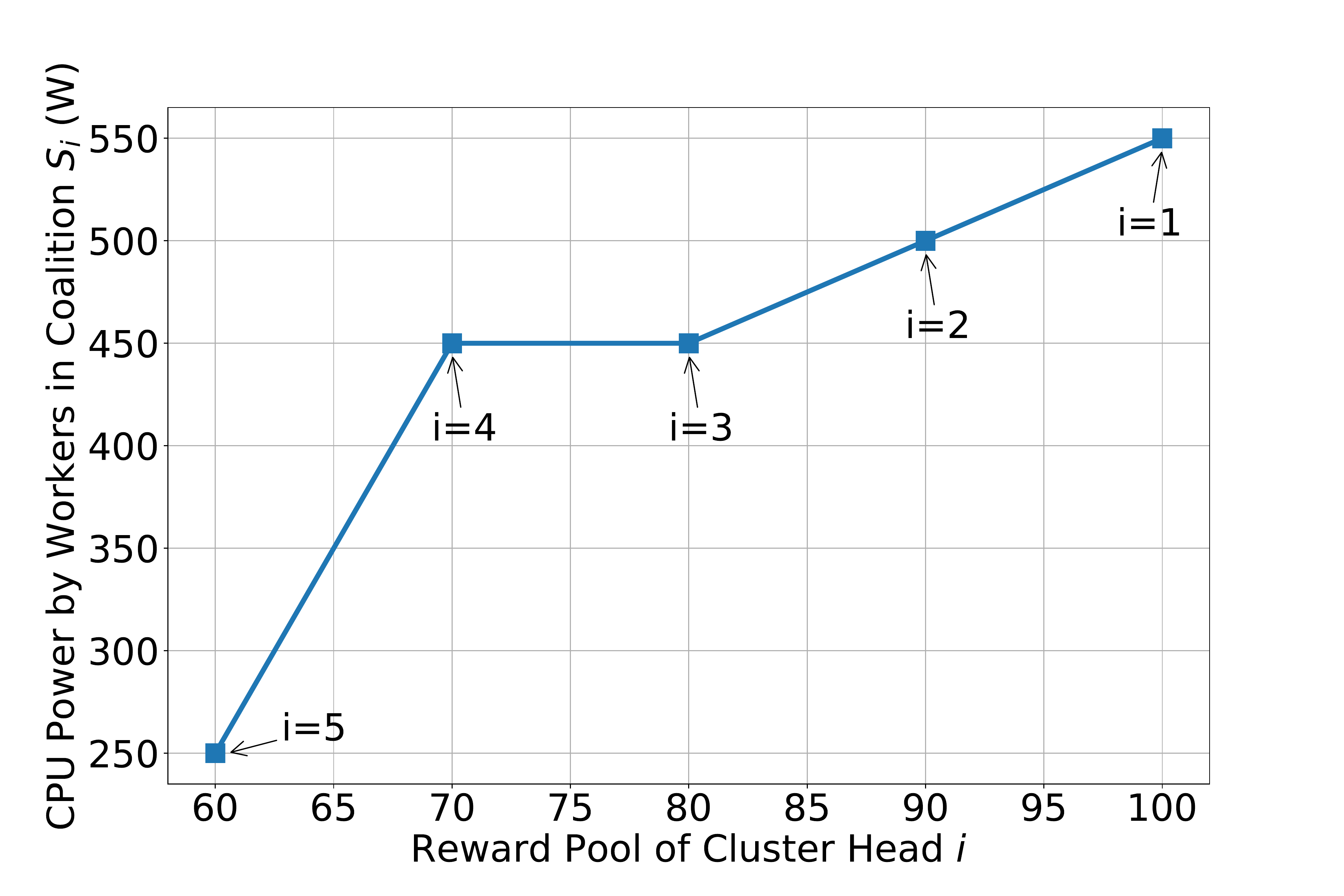}
\caption{CPU power by workers in coalition $S_i$ vs the reward pool offered by cluster head $i$.}
\label{fig:reward}
\end{figure}

From Fig.~\ref{fig:reward}, we observe that as the amount of reward pool offered by a cluster head increases, the total amount of CPU power of the workers in the coalition increases. For example, cluster head~1 offers a reward of 100 and forms a coalition with worker~2 and worker~7 having CPU powers of $150$W and $400$W respectively.

\subsection{Upper-level All-pay Auction}

\subsubsection{Monotonic Behaviour of Workers}

In the simulations, we consider a uniform distribution of the cluster heads' valuation for the rewards, where $v_i \in [0,1]$ which are independently drawn from $F(v)=v$. From Figs.~\ref{fig:onereward}-\ref{fig:diffrewards}, it can be observed that the cluster head's CPU power allocation increases monotonically with its valuation. Specifically, the higher the valuation of the cluster head for the rewards, the larger the amount of CPU power allocated for the CDC subtask. Since the cluster heads are symmetric where their valuations are drawn from the same distribution, the cluster heads with the same valuation contribute the same amount of CPU power.

\begin{figure*}
\centering
\begin{multicols}{3}

\includegraphics[width=\columnwidth]{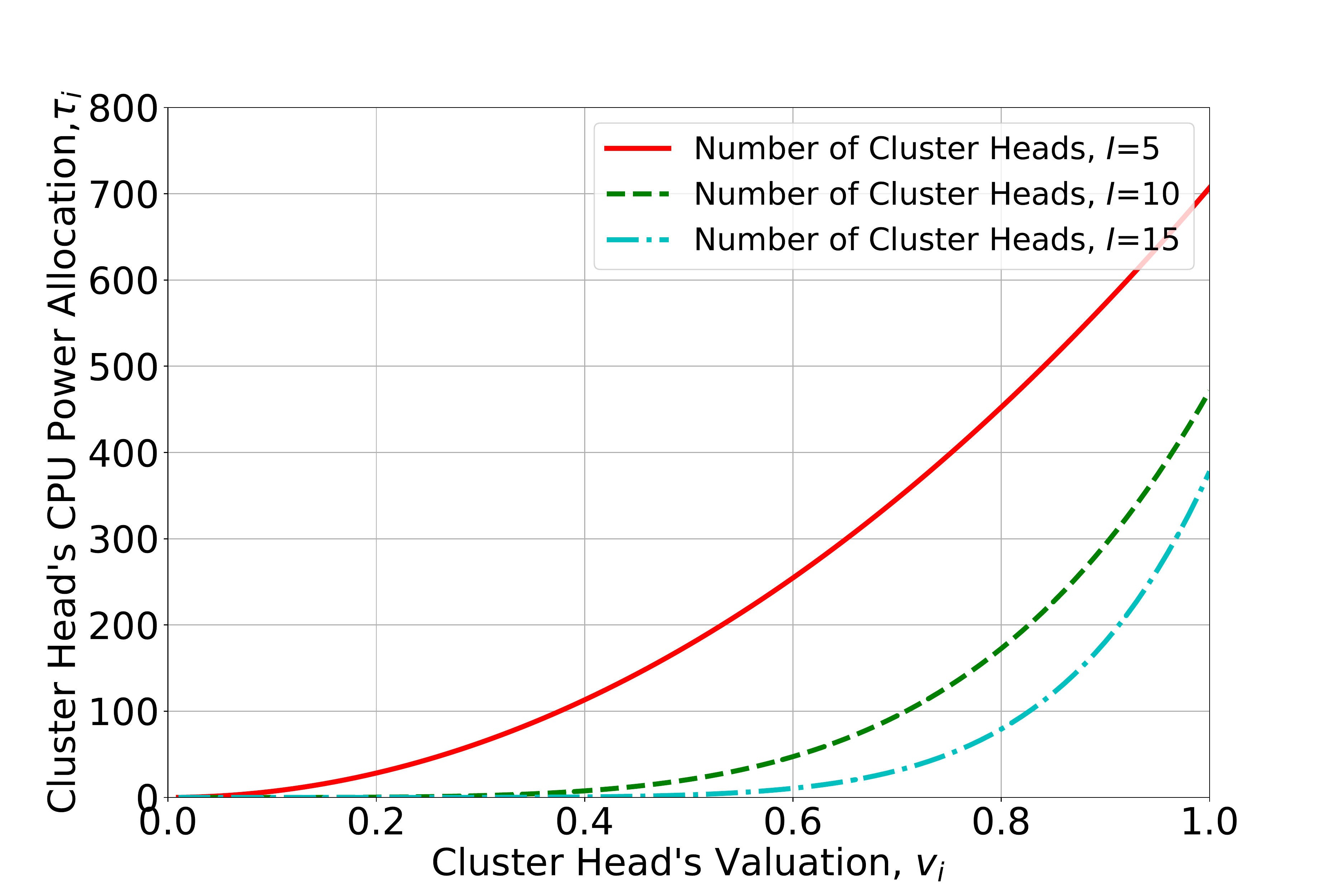}
	\caption{Only one reward is offered to the worker with the largest CPU power allocation.}
	\label{fig:onereward}

\includegraphics[width=\columnwidth]{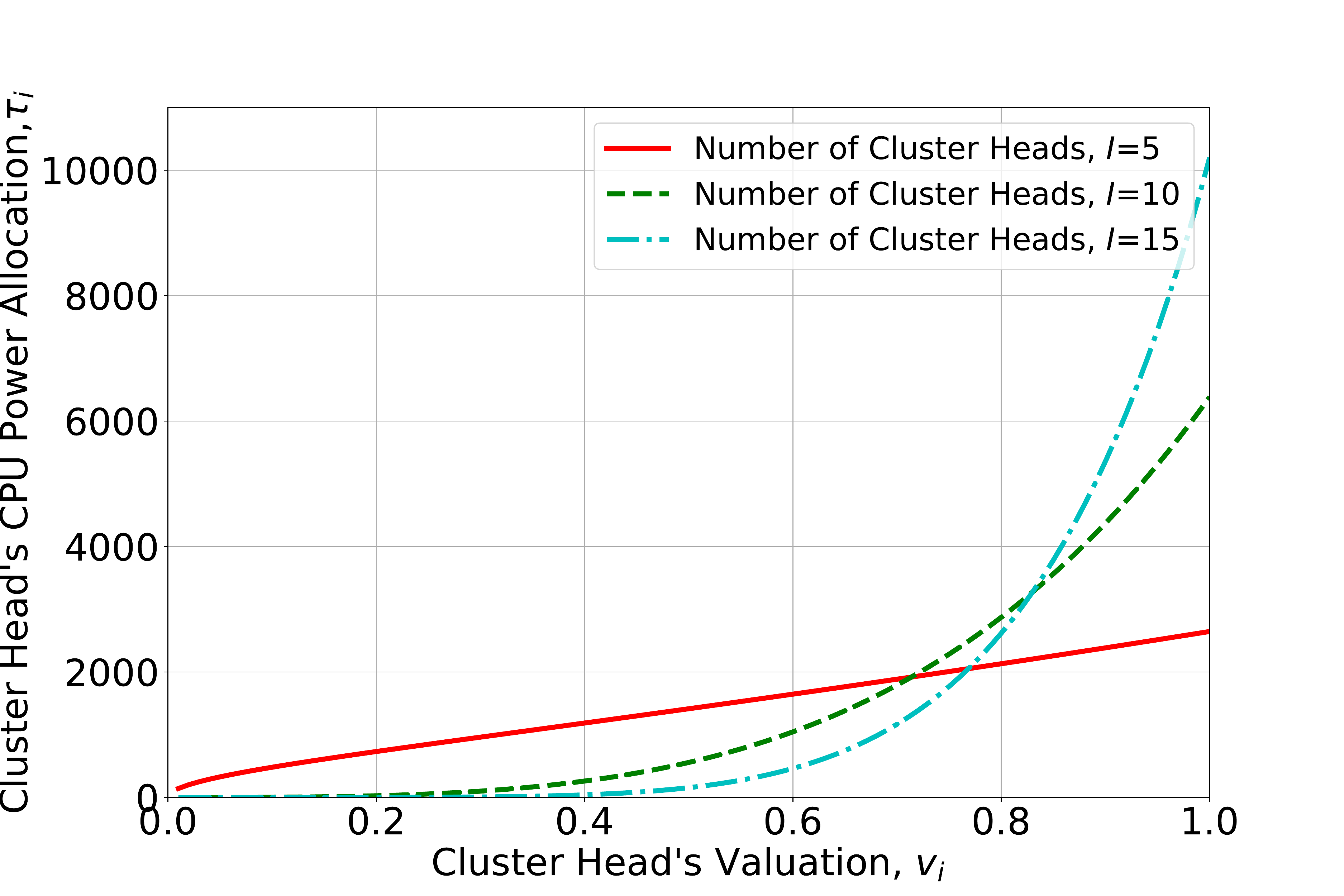}
	\caption{Homogeneous rewards, i.e., the difference between consecutive reward amounts is $0$.}
	\label{fig:homogeneous}

\includegraphics[width=\columnwidth]{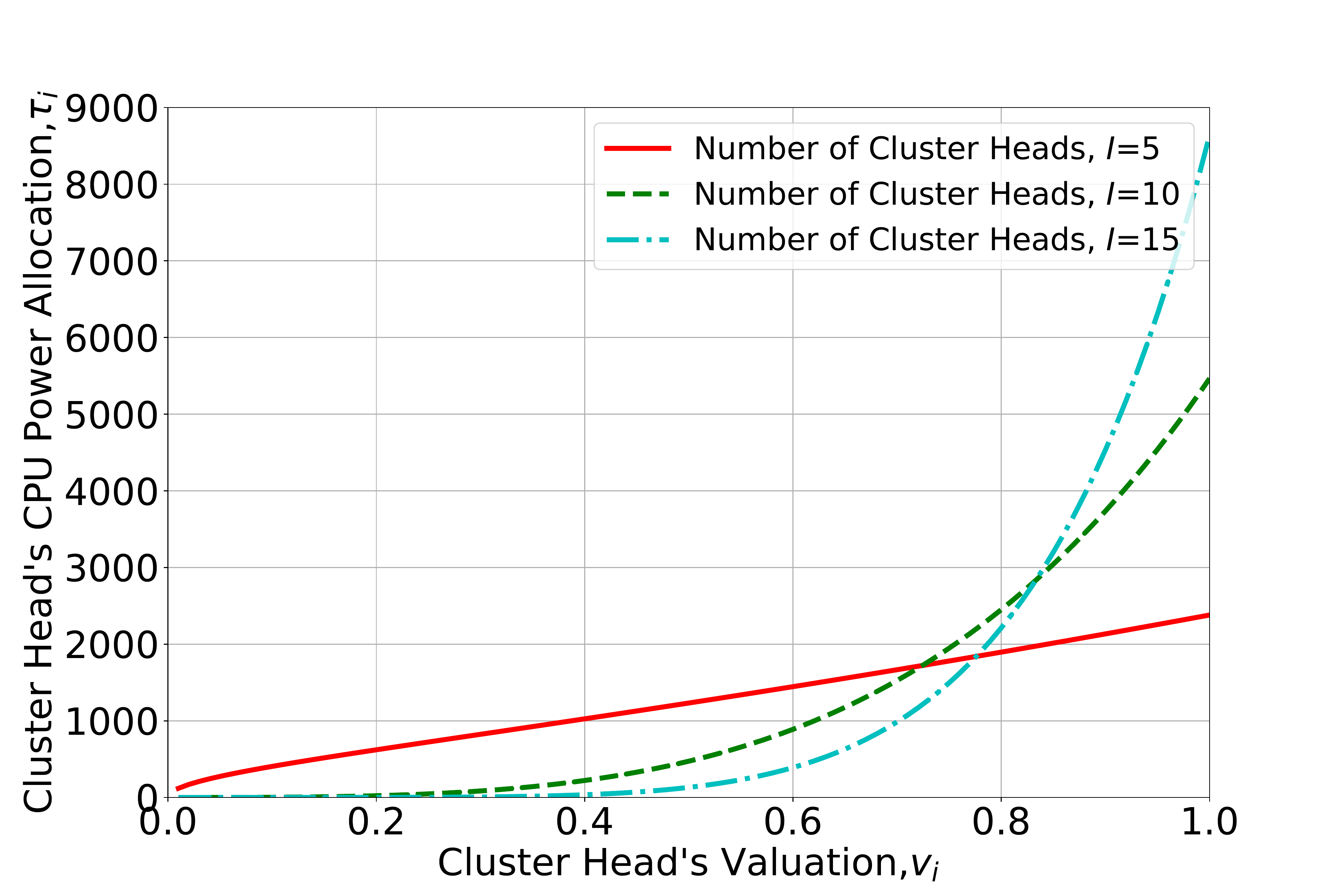}
    \caption{The difference between consecutive reward amounts is by a factor of 0.8, $M_{k+1}=0.8M_{k}$.}
    \label{fig:multigeom}

\end{multicols}
\end{figure*}

\begin{figure*}
\centering
\begin{multicols}{3}
\includegraphics[width=\columnwidth]{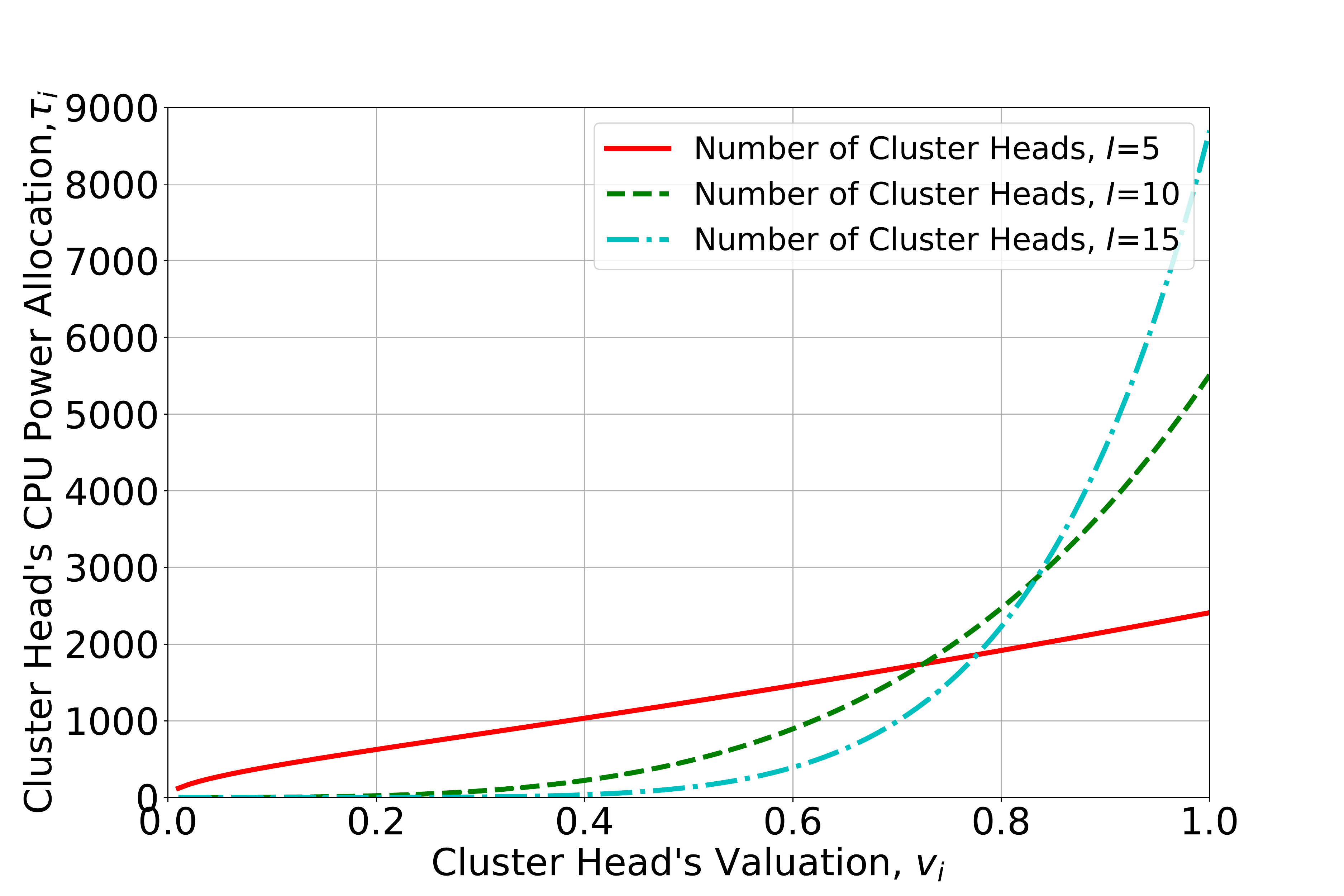} 
	\caption{The difference $M_{k}-M_{k+1}$ between consecutive reward amounts is 0.05.}
	\label{fig:multiarith0.05}
\includegraphics[width=\columnwidth]{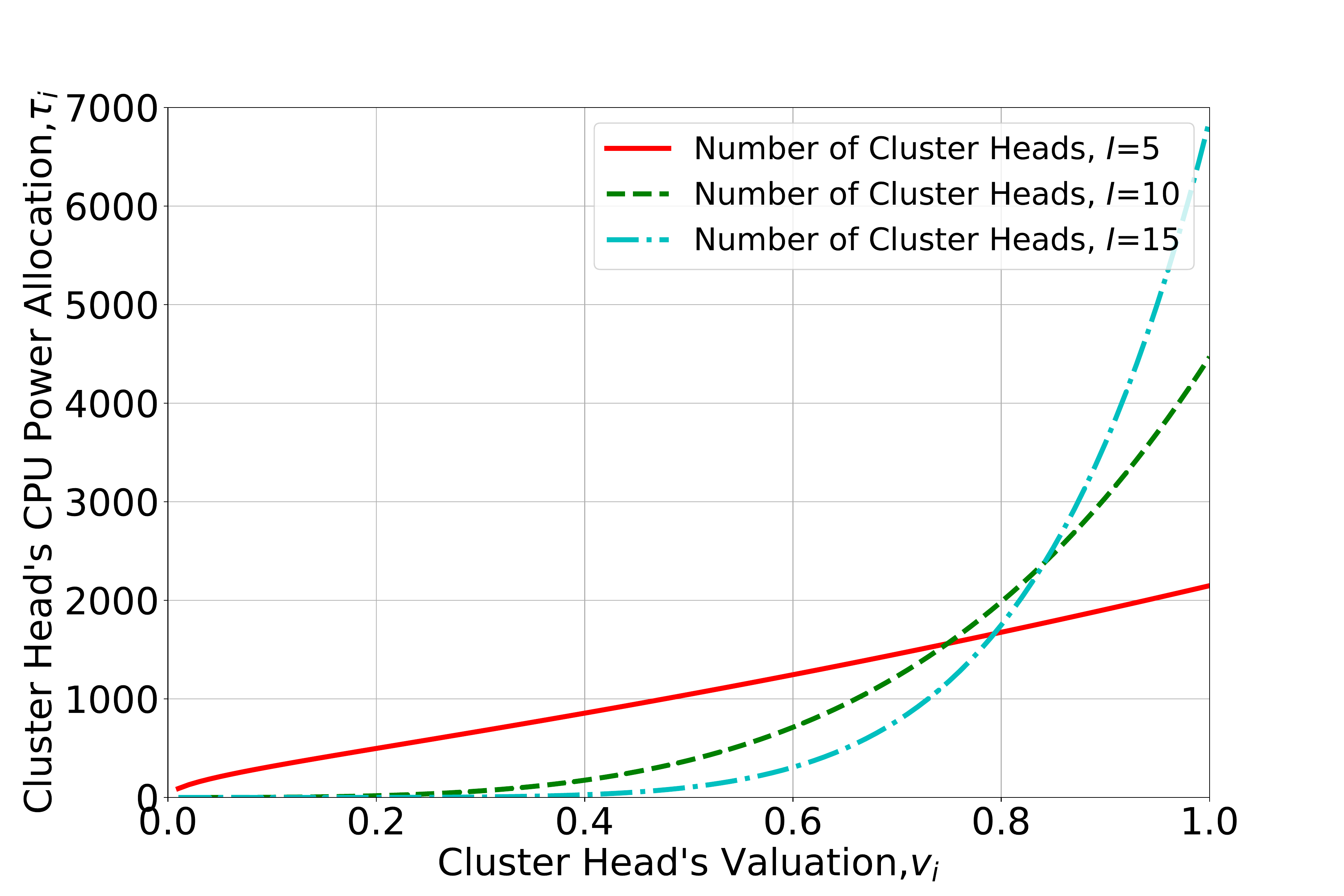} 
    	\caption{The difference $M_{k}-M_{k+1}$ between consecutive reward amounts is 0.1.}
    	\label{fig:multiarith0.1} 
\includegraphics[width=\columnwidth]{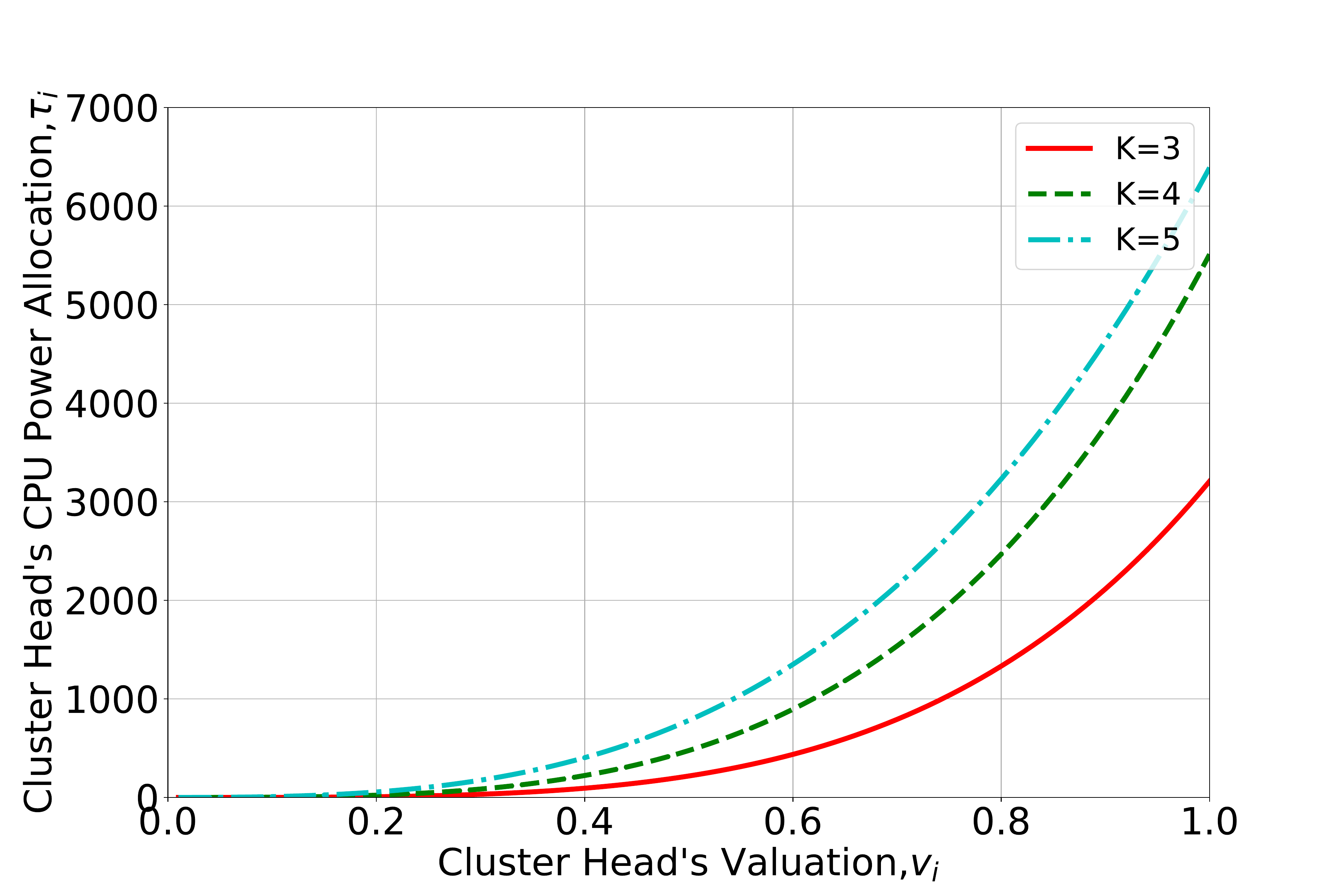}
	\caption{Different number $K$ of rewards, the difference between the consecutive rewards is 0.05, $I=10$.}
	\label{fig:diffrewards}

\end{multicols}
\end{figure*}

\subsubsection{Winner-take-all}

Based on the different reward structure adopted by the master, the cluster heads allocate their CPU power accordingly. Figure~\ref{fig:onereward} shows that when there are 5 workers and only one reward is offered to the cluster head that allocates the largest amount of CPU power, the cluster head with the highest valuation of 1, i.e., $v_{i}=1$, is only willing to contribute $707.1$W of CPU power. However, when the master offers multiple rewards, the cluster head with the same valuation of 1 is willing to contribute as high as $2380$W, $2410$W and $2149$W as shown in Fig.~\ref{fig:multigeom}, Fig.~\ref{fig:multiarith0.05} and Fig.~\ref{fig:multiarith0.1} respectively. With more rewards, the cluster heads have higher chance of winning one of the rewards. Hence, to incentivize the cluster heads to allocate more CPU power for the CDC subtasks, the master is better off offering multiple rewards than a single reward.

\subsubsection{Multiple Rewards: }The master needs to decide between homogeneous and heterogeneous reward allocation. Homogeneous rewards means the total amount of reward is split equally among the winning cluster heads whereas heterogeneous rewards are allocated based on the rank of the cluster heads where the amount of reward offered to the cluster head decreases as its rank increases.

\begin{itemize}

\item{\textbf{Homogeneous Reward Allocation: }}We observe similar trends in both homogeneous and heterogeneous reward allocation. Specifically, the cluster heads with lower valuations allocate more CPU power when there are fewer cluster heads in the network whereas the cluster heads with higher valuations allocate more CPU power when there are more cluster heads in the network. Generally, the cluster heads of both low and high valuations allocate more CPU power when homogeneous rewards are allocated. Figure~\ref{fig:homogeneous} shows that when there are 10 cluster heads in the network, a cluster head with valuation of $0.9$ allocates $4371$W, which is higher than $3735$W, $3765$W and $3041$W in Fig.~\ref{fig:multigeom}, Fig.~\ref{fig:multiarith0.05} and Fig.~\ref{fig:multiarith0.1} respectively. Similarly, in a network with 10 cluster heads, a cluster head with valuation of $0.2$ allocates $26.8$W when homogeneous rewards are allocated, which is also higher than $22.5$W, $22.6$W and $17.5$W when the differences between the consecutive rewards are a factor of 0.8, a constant of 0.05 and 0.1 respectively.

\item{\textbf{Heterogeneous Reward Allocation: }}Figure~\ref{fig:multigeom}, Fig.~\ref{fig:multiarith0.05} and Fig.~\ref{fig:multiarith0.1} show the allocation of CPU power by the cluster heads under arithmetic and geometric reward sequences. When the difference between the consecutive rewards is smaller, the cluster heads are willing to allocate more CPU power. For example, when the difference between the consecutive rewards is 0.05, i.e., $M_{k}-M_{k+1}=0.05$, $k=1,2,\ldots,K-1$, the cluster head with valuation of 1 allocates $8698$W when there are 15 cluster heads competing for 4 rewards. However, under the same setting of 15 cluster heads competing for 4 rewards, the cluster head with valuation of 1 is only willing to allocate $6875$W when the difference between the consecutive rewards is 0.1.
  
\end{itemize}

\subsubsection{Effects of Different System Parameter Values}
Apart from the different reward structures, the cluster heads also behave differently when the system parameter values, e.g., the number of cluster heads and the number of rewards, are changed.

\begin{itemize} 

\item{\textbf{More Cluster Heads: }}When there is only one reward offered to the cluster head that allocates the largest amount of CPU power, the cluster heads allocate more CPU power when there are 5 cluster heads than that of 15 cluster heads. For example, in Fig.~\ref{fig:onereward}, the cluster head with a valuation of $0.8$ allocates $452.5$W when there are 5 workers but only allocates $79.3$W for computation when there are 15 workers. When there are more cluster heads participating in the auction, the competition among the cluster heads is stiffer and the probability of winning the reward decreases. As a result, the cluster heads allocate a smaller amount of CPU power.

However, similar trends are only observed for cluster heads with low valuations, e.g., $v_{i}=0.6$, when multiple rewards are offered. When the number of cluster heads increases, the cluster heads with low valuations reduce their allocation of CPU power for the CDC subtasks. However, this is not observed for cluster heads with high valuations, e.g., $v_{i}=0.9$. Figure~\ref{fig:multigeom}, Fig.~\ref{fig:multiarith0.05} and Fig.~\ref{fig:multiarith0.1} show that the cluster heads with high valuations allocate more CPU power when there are more cluster heads competing for the multiple rewards offered by the master. Specifically, when the master offers 4 rewards with a difference of 0.05 between the consecutive rewards, the cluster head with a valuation of $0.9$ allocates $2159$W when there are 5 cluster heads but allocates $4563$W when there are 15 cluster heads. When multiple rewards are offered, since it is possible for the cluster heads to still win one of the remaining rewards even if they do not win the largest amount of reward, i.e., top reward, the cluster heads are more willing to allocate their CPU power for the CDC subtasks. %Although the competition among the workers increases when the number of workers participating in the all-pay auction increases, the workers with high valuations can still win some rewards even if they lose out in the competition for the top reward. 
Hence, the cluster heads with high valuations allocate more CPU power to increase their chance of winning the top reward. 

\item{\textbf{More Rewards: }} When the number of cluster heads participating in the all-pay auction is fixed, the cluster heads allocate more CPU power when there are more rewards that are offered. It is seen from Fig.~\ref{fig:diffrewards} that when there are 10 cluster heads in the all-pay auction, the cluster head with a valuation of $0.8$ allocates CPU power of $3231$W when 5 rewards are offered as compared to $1332$W and $2469$W when 3 and 4 rewards are offered respectively. 

\end{itemize}

\subsection{Comparison with Other Schemes}

\begin{figure}
\includegraphics[width=\linewidth]{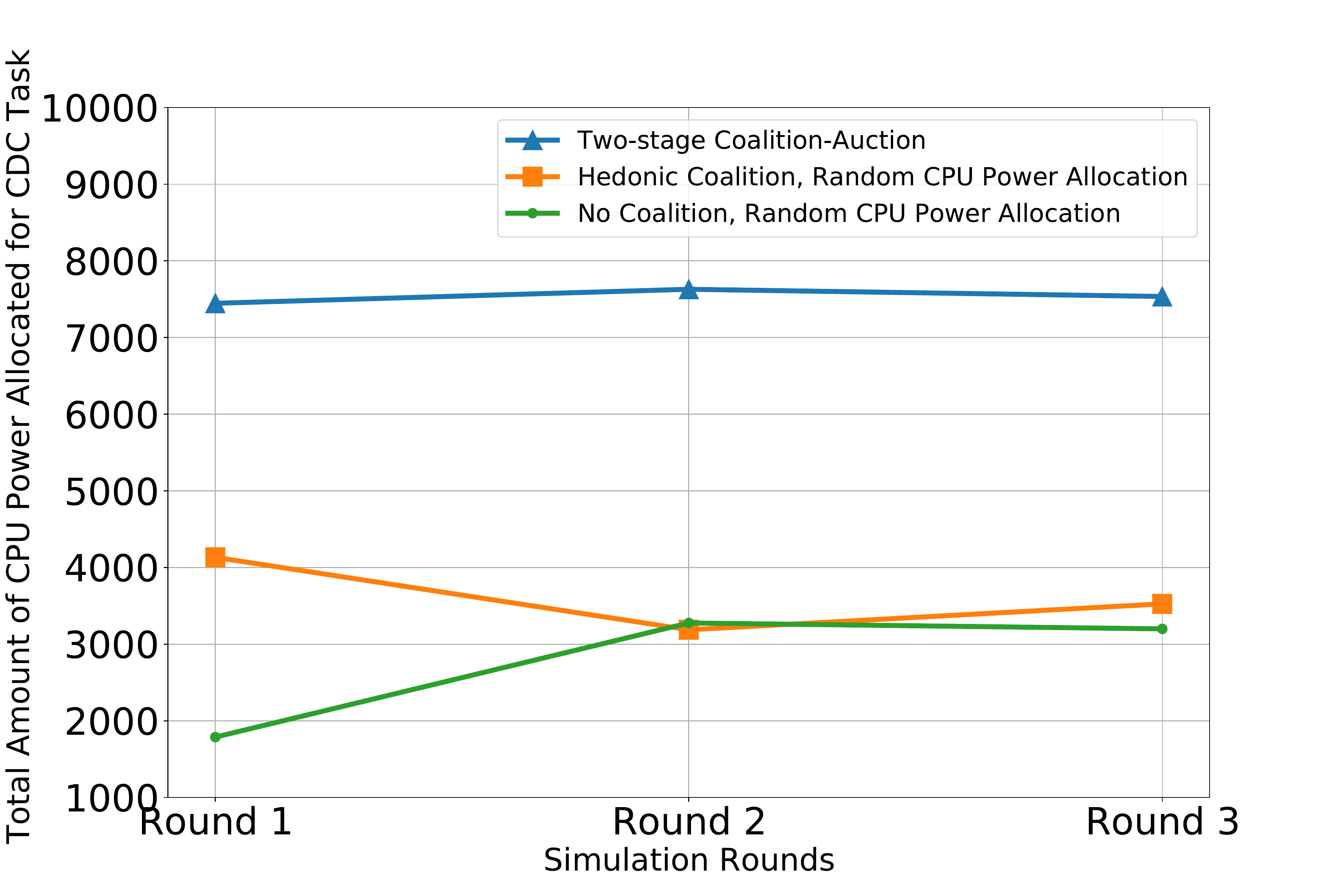}
	\caption{Comparison of the two-level coalition-auction approach with other schemes.}
	\label{fig:compare}
\end{figure}

We compare the proposed two-level coalition-auction approach against two other schemes, i.e., hedonic coalition formation among workers with random allocation of CPU power by the cluster heads and no coalition among workers with random allocation of CPU power by the cluster heads. Figure~\ref{fig:compare} shows the comparison of the performance of the proposed two-level game-theoretic approach against other schemes in a edge computing network with $5$ cluster heads. We observe that when coalitions are formed among workers, the total amount of CPU power allocated for the CDC subtasks is generally higher. This is because by forming coalitions with the workers, the cluster heads have more available CPU power to perform their computation tasks. Instead of randomly allocating CPU power for the CDC subtasks, the cluster heads are incentivized to allocate more CPU power when the master offers homogeneous rewards to the cluster heads under an all-pay auction. The average of the total amount of CPU power allocated for the CDC task under the two-level coalition-auction approach and the coalition with random allocation scheme are $7537$W and $3615$W respectively.

\section{Conclusion}
\label{sec:conclude}

In this paper, we proposed a two-level framework that incentivizes cluster heads and workers to contribute CPU power to facilitate the CDC tasks. The master applies the polynomial codes to divide the dataset and allocate the CDC subtasks to the cluster heads. In the lower level, we propose a hedonic coalition formation game in which each worker chooses its coalition based on its individual utility. In the upper level, we design an all-pay auction to incentivize cluster heads, given their coalitions of workers, to participate in the CDC tasks by contributing larger amount of CPU power. Specifically, we use the recovery threshold achieved by the polynomial codes to determine the number of rewards to be offered in the all-pay auction. Then, the master determines the reward structure to maximize its utility given the strategies of the cluster heads. Simulation results show that the utility of the cloud server is maximized when it offers multiple homogeneous rewards to incentivize the cluster heads to allocate more CDC power for the CDC subtasks.

In future work, we will consider workers with valuations chosen from probability distributions, other than uniform.

\bibliographystyle{ieeetr}
\bibliography{allpay_coalition}

\begin{thebibliography}{10}

\bibitem{liu2018survey}
J.~Liu, H.~Shen, H.~S. Narman, W.~Chung, and Z.~Lin, ``{A Survey of Mobile
  Crowdsensing Techniques: A Critical Component for The Internet of Things},''
  {\em ACM Transactions on Cyber-Physical Systems}, vol.~2, no.~3, 2018.

\bibitem{kalim2016crater}
F.~{Kalim}, J.~P. {Jeong}, and M.~U. {Ilyas}, ``{CRATER: A Crowd Sensing
  Application to Estimate Road Conditions},'' {\em IEEE Access}, vol.~4,
  pp.~8317--8326, 2016.

\bibitem{dutta2017towards}
J.~Dutta, C.~Chowdhury, S.~Roy, A.~I. Middya, and F.~Gazi, ``{Towards Smart
  City: Sensing Air Quality in City Based on Opportunistic Crowd-Sensing},'' in
  {\em Proceedings of the 18th International Conference on Distributed
  Computing and Networking}, ICDCN '17, (Hyderabad, India), Association for
  Computing Machinery, 2017.

\bibitem{jovanovic2019mobile}
S.~Jovanovi{\'c}, M.~Jovanovi{\'c}, T.~{\v{S}}kori{\'c}, S.~Joki{\'c},
  B.~Milovanovi{\'c}, K.~Katzis, and D.~Baji{\'c}, ``{A Mobile Crowd Sensing
  Application for Hypertensive Patients},'' {\em Sensors}, vol.~19, no.~2,
  p.~400, 2019.

\bibitem{shi2016promise}
W.~{Shi} and S.~{Dustdar}, ``{The Promise of Edge Computing},'' {\em Computer},
  vol.~49, no.~5, pp.~78--81, 2016.

\bibitem{dean2013tail}
J.~Dean and L.~A. Barroso, ``{The Tail at Scale},'' {\em Communications of the
  ACM}, vol.~56, no.~2, p.~74–80, 2013.

\bibitem{ananthanarayanan2010reining}
G.~Ananthanarayanan, S.~Kandula, A.~G. Greenberg, I.~Stoica, Y.~Lu, B.~Saha,
  and E.~Harris, ``{Reining in the Outliers in Map-Reduce Clusters using
  Mantri.},'' in {\em 9th USENIX Symposium on Operating Systems Design and
  Implementation (OSDI 2010)}, p.~24, 2010.

\bibitem{ng2020survey}
J.~S. Ng, W.~Y.~B. Lim, N.~C. Luong, Z.~Xiong, A.~Asheralieva, D.~Niyato,
  C.~Leung, and C.~Miao, ``{A Survey of Coded Distributed Computing},'' {\em
  arXiv preprint arXiv:2008.09048}, 2020.

\bibitem{aktas2018relaunch}
M.~F. Aktas, P.~Peng, and E.~Soljanin, ``{Straggler Mitigation by Delayed
  Relaunch of Tasks},'' {\em ACM SIGMETRICS Performance Evaluation Review},
  vol.~45, no.~3, pp.~224--231, 2018.

\bibitem{aktas2017clones}
M.~F. Aktas, P.~Peng, and E.~Soljanin, ``{Effective Straggler Mitigation: Which
  Clones Should Attack and When?},'' {\em ACM SIGMETRICS Performance Evaluation
  Review}, vol.~45, no.~2, p.~12–14, 2017.

\bibitem{maddah2014fundamental}
M.~A. Maddah-Ali and U.~Niesen, ``{Fundamental Limits of Caching},'' {\em IEEE
  Transactions on Information Theory}, vol.~60, no.~5, pp.~2856--2867, 2014.

\bibitem{karam2016hierarchical}
N.~{Karamchandani}, U.~{Niesen}, M.~A. {Maddah-Ali}, and S.~N. {Diggavi},
  ``{Hierarchical Coded Caching},'' {\em IEEE Transactions on Information
  Theory}, vol.~62, no.~6, pp.~3212--3229, 2016.

\bibitem{zhang2013performance}
Z.~{Zhang}, L.~{Cherkasova}, and B.~T. {Loo}, ``{Performance Modeling of
  MapReduce Jobs in Heterogeneous Cloud Environments},'' in {\em 2013 IEEE
  Sixth International Conference on Cloud Computing}, (Santa Clara, California,
  USA), pp.~839--846, 2013.

\bibitem{li2018tradeoff}
S.~{Li}, M.~A. {Maddah-Ali}, Q.~{Yu}, and A.~S. {Avestimehr}, ``{A Fundamental
  Tradeoff Between Computation and Communication in Distributed Computing},''
  {\em IEEE Transactions on Information Theory}, vol.~64, no.~1, pp.~109--128,
  2018.

\bibitem{haddadpour2018cross}
F.~{Haddadpour}, Y.~{Yang}, V.~{Cadambe}, and P.~{Grover}, ``{Cross-Iteration
  Coded Computing},'' in {\em 2018 56th Annual Allerton Conference on
  Communication, Control, and Computing (Allerton)}, (Monticello, Illinois,
  USA), pp.~196--203, 2018.

\bibitem{wan2020topological}
K.~Wan, M.~Ji, and G.~Caire, ``{Topological Coded Distributed Computing},''
  {\em arXiv preprint arXiv:2004.04421}, 2020.

\bibitem{yu2017polynomial}
Q.~Yu, M.~Maddah-Ali, and S.~Avestimehr, ``{Polynomial Codes: an Optimal Design
  for High-Dimensional Coded Matrix Multiplication},'' in {\em Advances in
  Neural Information Processing Systems 30 (NIPS 2017)} (I.~Guyon, U.~V.
  Luxburg, S.~Bengio, H.~Wallach, R.~Fergus, S.~Vishwanathan, and R.~Garnett,
  eds.), pp.~4403--4413, Curran Associates, Inc., 2017.

\bibitem{lee2017high}
K.~{Lee}, C.~{Suh}, and K.~{Ramchandran}, ``{High-dimensional Coded Matrix
  Multiplication},'' in {\em 2017 IEEE International Symposium on Information
  Theory (ISIT)}, (Aachen, Germany), pp.~2418--2422, 2017.

\bibitem{raviv2019gradient}
N.~Raviv, R.~Tandon, A.~Dimakis, and I.~Tamo, ``{Gradient Coding from Cyclic
  {MDS} Codes and Expander Graphs},'' in {\em Proceedings of the 35th
  International Conference on Machine Learning} (J.~Dy and A.~Krause, eds.),
  vol.~80 of {\em Proceedings of Machine Learning Research},
  (Stockholmsmässan, Stockholm Sweden), pp.~4305--4313, PMLR, 2018.

\bibitem{dutta2017coded}
S.~{Dutta}, V.~{Cadambe}, and P.~{Grover}, ``{Coded Convolution for Parallel
  and Distributed Computing within a Deadline},'' in {\em 2017 IEEE
  International Symposium on Information Theory (ISIT)}, (Aachen, Germany),
  pp.~2403--2407, 2017.

\bibitem{wang2018fundamental}
S.~Wang, J.~Liu, N.~Shroff, and P.~Yang, ``{Fundamental Limits of Coded Linear
  Transform},'' {\em arXiv preprint arXiv:1804.09791}, 2018.

\bibitem{yu2017fourier}
Q.~{Yu}, M.~A. {Maddah-Ali}, and A.~S. {Avestimehr}, ``{Coded Fourier
  Transform},'' in {\em 2017 55th Annual Allerton Conference on Communication,
  Control, and Computing (Allerton)}, (Monticello, Illinois, USA),
  pp.~494--501, 2017.

\bibitem{kiani2018exploitation}
S.~{Kiani}, N.~{Ferdinand}, and S.~C. {Draper}, ``{Exploitation of Stragglers
  in Coded Computation},'' in {\em 2018 IEEE International Symposium on
  Information Theory (ISIT)}, (Vail, Colorado, USA), pp.~1988--1992, 2018.

\bibitem{ozfatura2019speeding}
E.~{Ozfatura}, D.~{Gündüz}, and S.~{Ulukus}, ``{Speeding Up Distributed
  Gradient Descent by Utilizing Non-persistent Stragglers},'' in {\em 2019 IEEE
  International Symposium on Information Theory (ISIT)}, (Paris, France),
  pp.~2729--2733, 2019.

\bibitem{dai2020sazd}
M.~{Dai}, Z.~{Zheng}, S.~{Zhang}, H.~{Wang}, and X.~{Lin}, ``{SAZD: A Low
  Computational Load Coded Distributed Computing Framework for IoT Systems},''
  {\em IEEE Internet of Things Journal}, vol.~7, no.~4, pp.~3640--3649, 2020.

\bibitem{ennya2018computing}
Z.~{Ennya}, M.~Y. {Hadi}, and A.~{Abouaomar}, ``{Computing Tasks Distribution
  in Fog Computing: Coalition Game Model},'' in {\em 2018 6th International
  Conference on Wireless Networks and Mobile Communications (WINCOM)},
  (Marrakesh, Morocco), pp.~1--4, 2018.

\bibitem{zhang2018data}
T.~{Zhang}, ``{Data Offloading in Mobile Edge Computing: A Coalition and
  Pricing Based Approach},'' {\em IEEE Access}, vol.~6, pp.~2760--2767, 2018.

\bibitem{wahab2018trust}
O.~A. {Wahab}, J.~{Bentahar}, H.~{Otrok}, and A.~{Mourad}, ``{Towards
  Trustworthy Multi-Cloud Services Communities: A Trust-Based Hedonic
  Coalitional Game},'' {\em IEEE Transactions on Services Computing}, vol.~11,
  no.~1, pp.~184--201, 2018.

\bibitem{xu2020dynamic}
C.~{Xu}, K.~{Zhu}, R.~{Wang}, and Y.~{Xu}, ``{Dynamic Selection of Mining Pool
  with Different Reward Sharing Strategy in Blockchain Networks},'' in {\em ICC
  2020 - 2020 IEEE International Conference on Communications (ICC)}, (Dublin,
  Ireland), pp.~1--6, 2020.

\bibitem{archak2009optimal}
N.~Archak and A.~Sundararajan, ``{Optimal Design of Crowdsourcing Contests},''
  {\em International Conference on Information Systems (ICIS) 2009
  Proceedings}, p.~200, 2009.

\bibitem{yoon2012optimal}
K.~Yoon, ``{The Optimal Allocation of Prizes in Contests: An Auction
  Approach},'' tech. rep., 2012.

\bibitem{ng2020collaborative}
J.~S. Ng, W.~Y.~B. Lim, S.~Garg, Z.~Xiong, D.~Niyato, M.~Guizani, and C.~Leung,
  ``{Collaborative Coded Computation Offloading: An All-pay Auction
  Approach},'' {\em arXiv preprint arXiv:2012.04854}, 2020.

\bibitem{xiao2016asymmetric}
J.~Xiao, ``{Asymmetric All-pay Contests with Heterogeneous Prizes},'' {\em
  Journal of Economic Theory}, vol.~163, pp.~178 -- 221, 2016.

\bibitem{tie2014optimal}
{Tie Luo}, S.~S. {Kanhere}, S.~K. {Das}, and {Hwee-Pink Tan}, ``{Optimal Prizes
  for All-Pay Contests in Heterogeneous Crowdsourcing},'' in {\em 2014 IEEE
  11th International Conference on Mobile Ad Hoc and Sensor Systems},
  (Philadelphia, Pennsylvania, USA), pp.~136--144, 2014.

\bibitem{cohen2008allocation}
C.~Cohen and A.~Sela, ``{Allocation of Prizes in Asymmetric All-pay
  Auctions},'' {\em European Journal of Political Economy}, vol.~24, no.~1,
  pp.~123 -- 132, 2008.

\bibitem{wen2016optimal}
Z.~Wen and L.~Lin, ``{Optimal Fee Structures of Crowdsourcing Platforms},''
  {\em Decision Sciences}, vol.~47, no.~5, pp.~820--850, 2016.

\bibitem{didier2009efficient}
F.~Didier, ``{Efficient Erasure Decoding of Reed-Solomon Codes},'' {\em arXiv
  preprint arXiv:0901.1886}, 2009.

\bibitem{luo2016incentive}
T.~Luo, S.~K. Das, H.~P. Tan, and L.~Xia, ``{Incentive Mechanism Design for
  Crowdsourcing: An All-Pay Auction Approach},'' {\em ACM Transactions on
  Intelligent Systems and Technology}, vol.~7, no.~3, 2016.

\bibitem{saad2009coalitional}
W.~{Saad}, Z.~{Han}, M.~{Debbah}, A.~{Hjorungnes}, and T.~{Basar},
  ``{Coalitional Game Theory for Communication Networks},'' {\em IEEE Signal
  Processing Magazine}, vol.~26, no.~5, pp.~77--97, 2009.

\bibitem{anna2002stability}
A.~Bogomolnaia and M.~O. Jackson, ``{The Stability of Hedonic Coalition
  Structures},'' {\em Games and Economic Behavior}, vol.~38, no.~2, pp.~201 --
  230, 2002.

\bibitem{zhang2018energy}
Y.~{Zhang}, J.~{He}, and S.~{Guo}, ``{Energy-Efficient Dynamic Task Offloading
  for Energy Harvesting Mobile Cloud Computing},'' in {\em 2018 IEEE
  International Conference on Networking, Architecture and Storage (NAS)},
  (Chongqing, China), pp.~1--4, 2018.

\bibitem{hao2018energy}
Y.~{Hao}, M.~{Chen}, L.~{Hu}, M.~S. {Hossain}, and A.~{Ghoneim}, ``{Energy
  Efficient Task Caching and Offloading for Mobile Edge Computing},'' {\em IEEE
  Access}, vol.~6, pp.~11365--11373, 2018.

\end{thebibliography}

\begin{IEEEbiography}[{\includegraphics[width=1in,height=1.25in,clip,keepaspectratio]{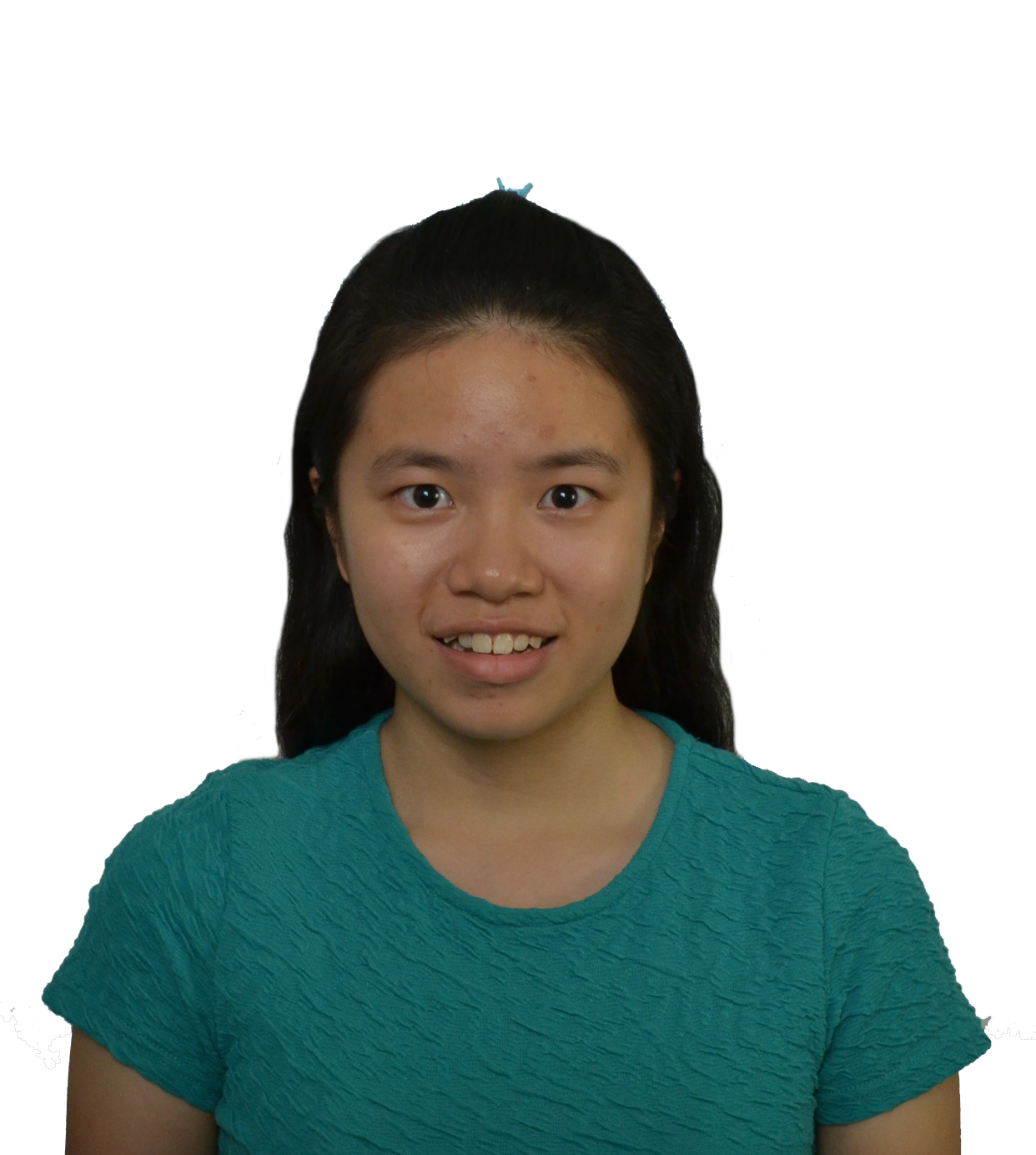}}]{Jer Shyuang Ng} graduated with Double (Honours) Degree in Electrical Engineering (Highest Distinction) and Economics from National University of Singapore (NUS) in 2019. She is currently an Alibaba PhD candidate with the Alibaba Group and Alibaba-NTU Joint Research Institute, Nanyang Technological University (NTU), Singapore. Her research interests include incentive mechanisms and edge computing.
\end{IEEEbiography}

\vspace{-1.5cm}

\begin{IEEEbiography}[{\includegraphics[width=1in,height=1.25in,clip,keepaspectratio]{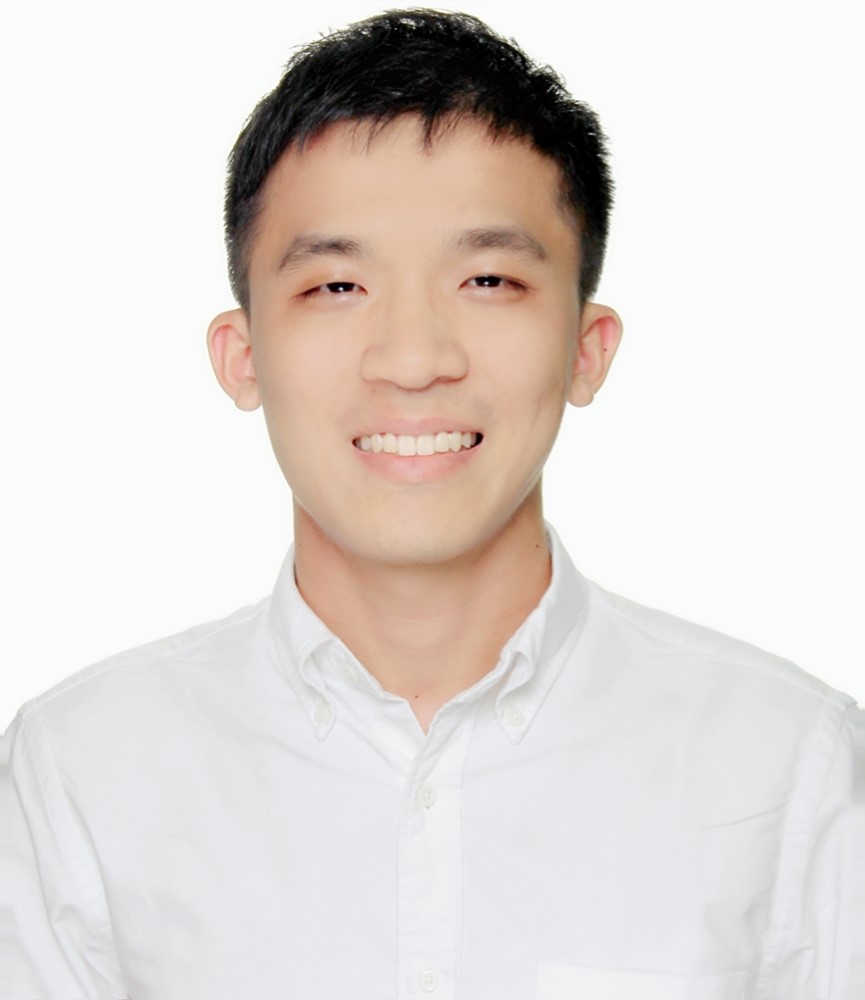}}]{Wei Yang Bryan Lim} graduated with double First Class Honours in Economics and Business Administration (Finance) from the National University of Singapore (NUS) in 2018. He is currently an Alibaba PhD candidate with the Alibaba Group and Alibaba-NTU Joint Research Institute, Nanyang Technological University, Singapore. His research interests include Federated Learning and Edge Intelligence.
\end{IEEEbiography}

\vspace{-1.5cm}

\begin{IEEEbiography}[{\includegraphics[width=1in,height=1.25in,clip,keepaspectratio]{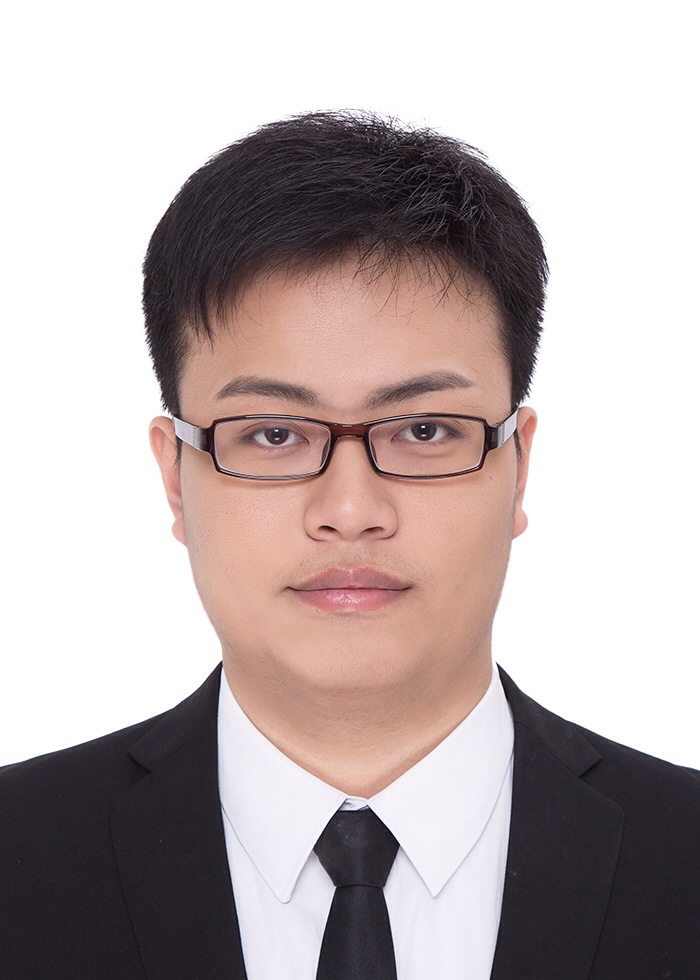}}]{Zehui Xiong}(S'17) received his B.Eng degree with the highest honors in Telecommunication Engineering from Huazhong University of Science and Technology, Wuhan, China, in Jul 2016. From Aug 2016 to Oct 2019, he pursued the Ph.D. degree in the School of Computer Science and Engineering, Nanyang Technological University, Singapore. Since Nov 2019, he has been with Alibaba-NTU Singapore Joint Research Institute. He was a visiting scholar with Department of Electrical Engineering at Princeton University from Jul to Aug 2019. He was also a visiting scholar with BBCR lab in Department of Electrical and Computer Engineering at University of Waterloo from Dec 2019 to Jan 2020. His research interests include resource allocation in wireless communications, network games and economics, blockchain, and edge intelligence.
\end{IEEEbiography}

\vspace{-1.5cm}

\begin{IEEEbiography}[{\includegraphics[width=1in,height=1.25in,clip,keepaspectratio]{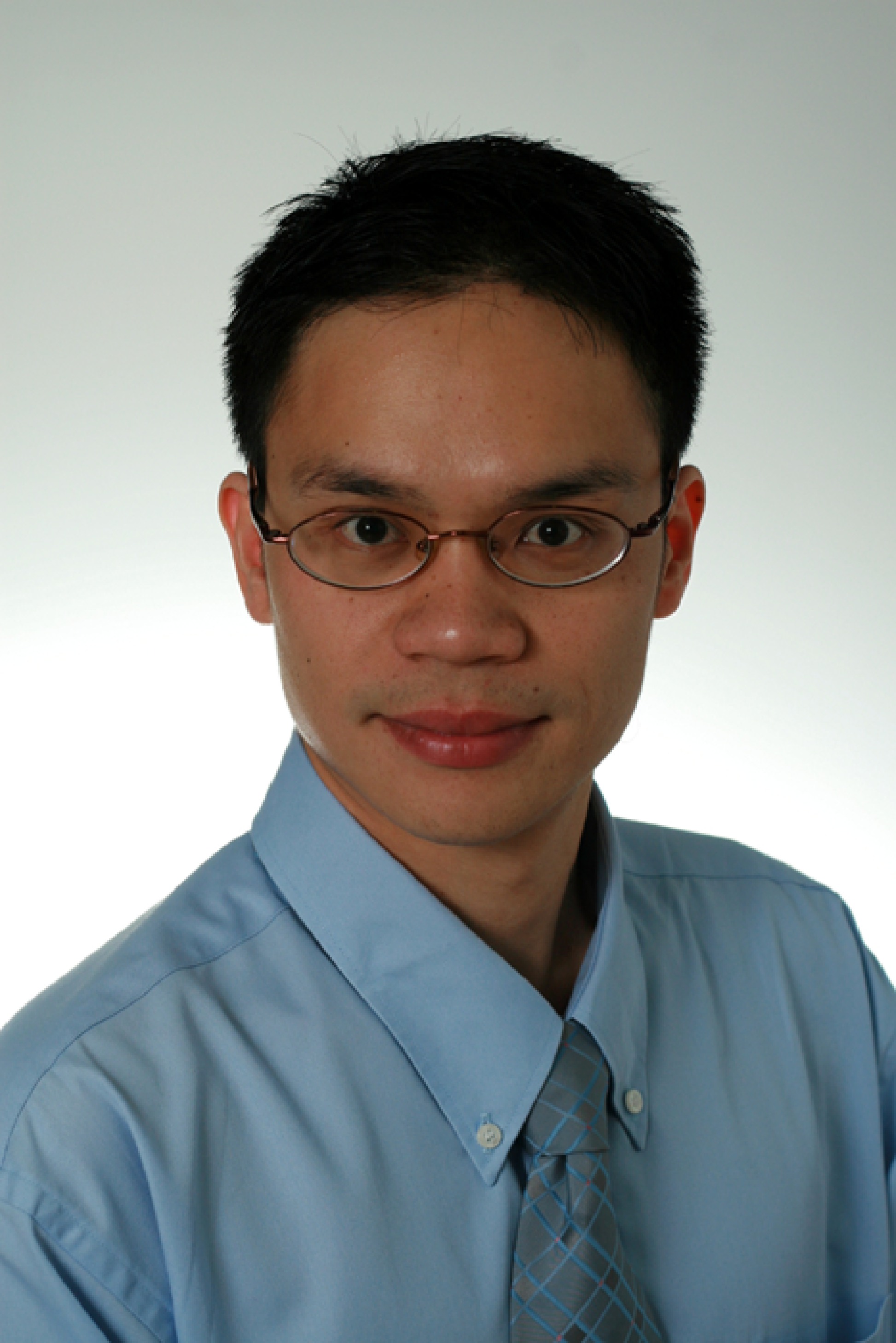}}]{Dusit Niyato} (M'09-SM'15-F'17) is currently a professor in the School of Computer Science and Engineering and, by courtesy, School of Physical \& Mathematical Sciences, at the Nanyang Technological University, Singapore. He received B.E. from King Mongkuk’s Institute of Technology Ladkrabang (KMITL), Thailand in 1999 and Ph.D. in Electrical and Computer Engineering from the University of Manitoba, Canada in 2008. He has published more than 380 technical papers in the area of wireless and mobile networking, and is an inventor of four US and German patents. He has authored four books including “Game Theory in Wireless and Communication Networks: Theory, Models, and Applications” with Cambridge University Press. He won the Best Young Researcher Award of IEEE Communications Society (ComSoc) Asia Pacific (AP) and The 2011 IEEE Communications Society Fred W. Ellersick Prize Paper Award. Currently, he is serving as a senior editor of IEEE Wireless Communications Letter, an area editor of IEEE Transactions on Wireless Communications (Radio Management and Multiple Access), an area editor of IEEE Communications Surveys and Tutorials (Network and Service Management and Green Communication), an editor of IEEE Transactions on Communications, an associate editor of IEEE Transactions on Mobile Computing, IEEE Transactions on Vehicular Technology, and IEEE Transactions on Cognitive Communications and Networking. He was a guest editor of IEEE Journal on Selected Areas on Communications. He was a Distinguished Lecturer of the IEEE Communications Society for 2016-2017. He was named the 2017, 2018, 2019 highly cited researcher in computer science. He is a Fellow of IEEE.
\end{IEEEbiography}

\vspace{-1.5cm}

\begin{IEEEbiography}[{\includegraphics[width=1in,height=1.25in,clip,keepaspectratio]{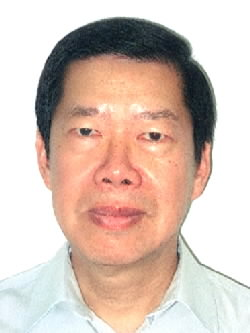}}]{Cyril Leung} received the B.Sc. (First Class Hons.) degree from Imperial College, University of London, U.K., and the M.S. and Ph.D. degrees in electrical engineering from Stanford University. He has been an Assistant Professor with the Department of Electrical Engineering and Computer Science, Massachusetts Institute of Technology, and the Department of Systems Engineering and Computing Science, Carleton University. Since 1980, he has been with the Department of Electrical and Computer Engineering, University of British Columbia (UBC), Vancouver, Canada, where he is a Professor and currently holds the PMC-Sierra Professorship in Networking and Communications. He served as an Associate Dean of Research and Graduate Studies with the Faculty of Applied Science, UBC, from 2008 to 2011. His research interests include wireless communication systems, data security and technologies to support ageless aging for the elderly. He is a member of the Association of Professional Engineers and Geoscientists of British Columbia, Canada.
\end{IEEEbiography}

\vspace{-1.5cm}

\begin{IEEEbiography}[{\includegraphics[width=1in,height=1.25in,clip,keepaspectratio]{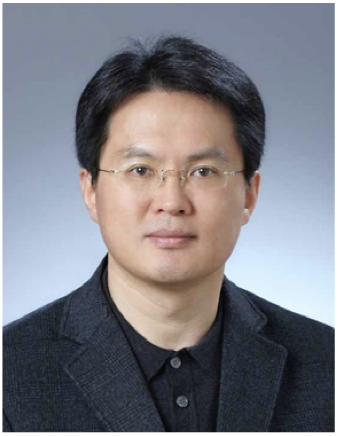}}]{Dong In Kim}(S’89–M’91–SM’02–F’19) received the Ph.D. degree in electrical engineering from the University of Southern California at Los Angeles, Los Angeles, CA, USA, in 1990. He was a tenured Professor with the School of Engi- neering Science, Simon Fraser University, Burn- aby, BC, Canada. Since 2007, he has been with Sungkyunkwan University, Suwon, South Korea, where he is currently a Professor with the College of Information and Communication Engineering.
He has been elevated to the grade of Fellow of the IEEE for his contributions to the cross-layer design of wireless communications systems. He is also a Fellow of the Korean Academy of Science and Technology and the National Academy of Engineering of Korea. He is a first recipient of the NRF of Korea Engineering Research Center in Wireless Communications for RF Energy Harvesting, from 2014 to 2021. From 2001 to 2014, he served as an Editor of Spread Spectrum Transmission and Access for the IEEE Transactions on Communications. From 2002 to 2011, he also served as an Editor and a Founding Area Editor of Cross-Layer Design and Optimization for the IEEE Transactions on Wireless Communications. From 2008 to 2011, he served as the Co-Editor-in-Chief for the IEEE/KICS Journal of Communications and Networks. He served as the Founding Editor-in-Chief for the IEEE Wireless Communications Letters, from 2012 to 2015. Since 2015, he has been serving as an Editor-at-Large of Wireless Communication I for the IEEE Transactions on Communications.
\end{IEEEbiography}

%\vspace{-1.5cm}

\begin{IEEEbiography}[{\includegraphics[width=1in,height=1.25in,clip,keepaspectratio]{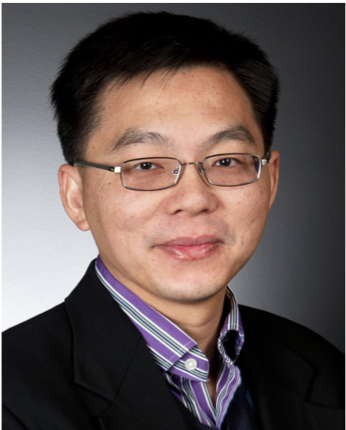}}]{Junshan Zhang} (S’98-M’00-SM’06-F’12) received the Ph.D. degree from the School of ECE, Purdue University, in 2000. He joined the School of ECEE, Arizona State University, AZ, USA, in 2000, where he has been a Professor since 2010. His research interests fall in the general field of information networks and its intersections with power networks and social networks, and fundamental problems in information networks and energy networks, including modeling and optimization for smart grid, optimization/control of mobile social networks and cognitive radio networks, and privacy/security in information networks. Dr. Zhang is a recipient of the ONR Young Investigator Award in 2005 and the NSF CAREER award in 2003. He received the Outstanding Research Award from the IEEE Phoenix Section in 2003. He co-authored two papers that received the Best Paper Runner-Up Award of the IEEE INFOCOM 2009 and the IEEE INFOCOM 2014, and a paper that received the IEEE ICC 2008 Best Paper Award. He was the TPC Co-Chair of a number of major conferences in communication networks, including INFOCOM 2012, WICON 2008, and IPCCC’06, and the TPC Vice-Chair of ICCCN’06. He was the General Chair of the IEEE Communication Theory Workshop 2007. He was an Associate Editor of the IEEE Transactions on Wireless Communications, an Editor of the Computer Network Journal, and an Editor the IEEE Wireless Communication Magazine. He was a Distinguished Lecturer of the IEEE Communications Society. He is currently serving as the Editor-in-Chief of the IEEE Transactions on Wireless Communications, an Editor-at-Large of the IEEE/ACM Transactions on Networking and an Editor of the IEEE Network Magazine.
\end{IEEEbiography}

\vspace{5.5cm}

\begin{IEEEbiography}[{\includegraphics[width=1in,height=1.25in,clip,keepaspectratio]{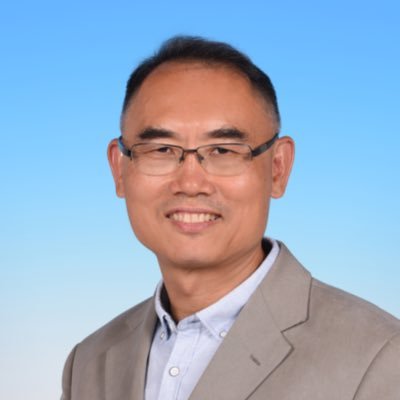}}]{Qiang Yang} is the head of AI at WeBank (Chief AI Officer) and Chair Professor at the Computer Science and Engineering (CSE) Department of Hong Kong University of Science and Technology (HKUST), where he was a former head of CSE Department and founding director of the Big Data Institute (2015-2018). His research interests include artificial intelligence, machine learning and data mining, especially in transfer learning, automated planning, federated learning and case-based reasoning. He is a fellow of several international societies, including ACM, AAAI, IEEE, IAPR and AAAS. He received his PhD from the Computer Science Department in 1989 and MSc in Astrophysics in 1985, both from the University of Maryland, College Park. He obtained his BSc in Astrophysics from Peking University in 1982.  He had been a faculty member at the University of Waterloo (1989-1995) and Simon Fraser University (1995-2001). He was the founding Editor in Chief of the ACM Transactions on Intelligent Systems and Technology (ACM TIST) and IEEE Transactions on Big Data (IEEE TBD).  He served as the President of International Joint Conference on AI (IJCAI, 2017-2019) and an executive council member of Association for the Advancement of AI (AAAI, 2016 - 2020).  Qiang Yang is a recipient of several awards, including the 2004/2005 ACM KDDCUP Championship, the ACM SIGKDD Distinguished Service Award (2017) and AAAI Innovative AI Applications Award (2016, 2019).  He was the founding director of Huawei's Noah's Ark Lab (2012-2014) and a co-founder of 4Paradigm Corp, an AI platform company. He is an author of several books including Transfer Learning (Cambridge Press), Federated Learning (Morgan Claypool), Intelligent Planning (Springer), Crafting Your Research Future (Morgan Claypool) and Constraint-based Design Recovery for Software Engineering (Springer).
\end{IEEEbiography}

\end{document}